\documentclass[12pt]{article}

\usepackage[a4paper,top=2.5cm,bottom=2.5cm,left=2.5cm,right=2.5cm,marginparwidth=1.75cm]{geometry}
\usepackage{lineno}
\usepackage{natbib}
\usepackage{amssymb}
\usepackage{bm}
\usepackage{authblk}
\usepackage{amsmath}
\usepackage{amsthm}
\usepackage{epsfig}
\usepackage{graphicx}
\usepackage{graphics}
\usepackage{float}
\usepackage{subfigure}
\usepackage{multirow}
\usepackage{color}
\usepackage{fullpage}
\usepackage[normalem]{ulem} 
\usepackage{makeidx}
\usepackage{xspace}
\usepackage{wrapfig}
\usepackage{setspace}
\usepackage{kotex}
\usepackage{url}
\usepackage{makecell}  
\makeindex

\newtheorem{proposition}{Proposition}
\newtheorem{theorem}{Theorem}

\newtheorem{lemma}{Lemma}

\newcommand{\xb}{{\bm{x}}}
\newcommand{\Xb}{{\bm{X}}}

\newcommand{\ub}{{\bm{u}}}
\newcommand{\Ub}{{\bm{U}}}

\newcommand{\betab}{{\bm{\beta}}}
\newcommand{\gammab}{{\bm{\gamma}}}

\newcommand{\alphab}{{\bm{\alpha}}}

\newcommand{\thetab}{{\bm{\theta}}}

\newcommand{\Gb}{{\bm{G}}}
\newcommand{\Thetab}{{\bm{\Theta}}}
\newcommand{\Omegab}{{\bm{\Omega}}}

\providecommand{\keywords}[1]
{
  \small	
  \textbf{\textit{Keywords: }} #1
}

\begin{document}

\title{Semiparametric robust mixture of experts based on nonparametric maximum likelihood}

\author[1]{Sangkon Oh}
\affil[1]{Department of Statistics and Data Science, Pukyong National University}

\author[2]{Victor H. Lachos}
\affil[2]{Department of Statistics, University of Connecticut}

\author[3]{Byungtae Seo}
\affil[3]{Department of Statistics, Sungkyunkwan University}

\date{}

\maketitle

\begin{abstract}
The mixture of experts (MoE) model provides a flexible approach for modeling heterogeneous regression relationships by allowing covariate-dependent mixing through a gating network, but most existing MoE models rely on parametric assumptions for expert error distributions, typically Gaussian, which can lead to inefficiency and sensitivity to outliers or heavy-tailed behavior when misspecified. We propose a semiparametric MoE model in which each expert error distribution is represented as a nonparametric Gaussian scale mixture estimated via nonparametric maximum likelihood, relaxing parametric assumptions within the Gaussian scale-mixture class while preserving the interpretability and structure of the MoE framework. The resulting model adapts to complex error structures, improves robustness under contamination and heavy tails, and remains competitive under well-specified Gaussian settings, providing a practical and theoretically grounded alternative to parametric MoE formulations.
\end{abstract}


\section{Introduction} \label{sec:intro}

The mixture of experts (MoE) model, originally introduced by \citet{jacobs1991adaptive}, provides a regression specification in which both component-specific models and mixing proportions depend on covariates. By incorporating a gating network that probabilistically assigns observations to latent experts based on predictors, the MoE framework extends classical finite mixture models and enables flexible modeling of heterogeneous data structures.

Subsequent developments have focused on improving clustering accuracy, predictive performance, and robustness. For example, \citet{nguyen2016laplace} considered Laplace-distributed expert errors as a robust alternative to Gaussian assumptions. Extensions based on heavy-tailed and asymmetric distributions were proposed by \citet{chamroukhi2016robust} and \citet{chamroukhi2017skew}, who developed $t$ and skew-$t$ MoE models to better accommodate non-Gaussian behavior. Related work has explored expert models based on scale mixtures of normal distributions to capture additional flexibility in error structure \citep{mirfarah2021mixture}. More recent contributions include MoE formulations with Gaussian location-scale mixtures for joint modeling of covariates and responses \citep{oh2023merging}, as well as partially linear expert models designed to increase structural flexibility \citep{hwang2025mixture}.
Despite these advances, most existing MoE approaches remain fundamentally parametric with respect to expert error distributions. While heavy-tailed or skewed specifications improve robustness relative to Gaussian models, they still impose rigid distributional forms that may be inadequate when the true error structure deviates substantially from assumed families. Residual misspecification can therefore limit both estimation stability and predictive performance.

To address this limitation, we introduce a semiparametric MoE approach in which each expert error distribution is modeled as an unknown Gaussian scale mixture. The mixing distributions are estimated via a nonparametric maximum likelihood estimator (NPMLE), allowing the error structure to adapt flexibly to the data without imposing restrictive parametric assumptions. The theoretical foundation of NPMLE traces back to the seminal work of \citet{kiefer1956consistency} and has been extensively developed in the literature \citep{Lindsay_1995, vvt96, mv00}. A growing body of research demonstrates the effectiveness of NPMLE-based modeling in diverse settings, including robust regression, heteroscedastic modeling, survival analysis, and mixture models \citep{SL15, XYS16, Seo_2017, AFT_NGSM, seo2024adaptive, lee2024finite, oh2024semiparametric, park2025penalized, Oh2026}.

The proposed method is motivated by the need for expert models that remain reliable in the presence of unknown or complex error structures. In practical applications, expert-specific error distributions are rarely known a priori and may exhibit heavy tails, contamination, or approximately Gaussian behavior. By learning each expert error distribution nonparametrically, the proposed approach provides adaptive robustness while preserving the interpretability and structural advantages of the MoE framework. We establish key theoretical properties of the resulting estimator, including identifiability, consistency, and monotonicity of the estimation algorithm, thereby establishing new theoretical guarantees for semiparametric mixture models. Extensive simulation studies and real data analyses further demonstrate that the proposed method achieves greater robustness and estimation stability than parametric MoE models, particularly in the presence of heavy-tailed errors or contamination.

The remainder of this paper is organized as follows. Section~\ref{sec:review} reviews the MoE framework, Gaussian scale mixture distributions, and nonparametric maximum likelihood estimation. Section~\ref{sec:proposed} presents the proposed semiparametric MoE model and its EM-based estimation procedure. Simulation studies are reported in Section~\ref{sec:sim}, followed by real data applications in Section~\ref{sec:real}. Section~\ref{sec:disc} concludes with a discussion and future research directions.

\section{Literature Review} \label{sec:review}
\subsection{Mixture of experts} 

Let $Y$ be a univariate response variable and let $\boldsymbol{X}=(1,\bm{U}^\top)^\top$ denote a $(p+1)$-dimensional covariate vector including an intercept.
Let $Z$ be a latent variable indicating the component membership of each observation.
Following \citet{jacobs1991adaptive}, the conditional distribution of $y$ given $\boldsymbol{x}$ is modeled as a finite mixture of regression models,
\begin{equation*}
p(y \mid \boldsymbol{x})
=
\sum_{k=1}^K
\pi(\boldsymbol{x};\boldsymbol{\alpha}_k)\,
\frac{1}{\sigma_k}
\phi\!\left(
\frac{y - \boldsymbol{x}^\top \boldsymbol{\beta}_k}{\sigma_k}
\right),
\end{equation*}
where $\pi(\boldsymbol{x};\boldsymbol{\alpha}_k)$ denotes a covariate-dependent mixing proportion satisfying
$0<\pi(\boldsymbol{x};\boldsymbol{\alpha}_k)<1$ and $\sum_{k=1}^K \pi(\boldsymbol{x};\boldsymbol{\alpha}_k)=1$.
Here, $\boldsymbol{\beta}_k$ is a $(p+1)$-dimensional vector of regression coefficients for the $k$th component,
$\sigma_k > 0$ is a component-specific scale parameter,
and $\phi(\cdot)$ denotes the standard normal density.
The mixing proportions are specified through a multinomial logistic function,
\[
\pi(\boldsymbol{x};\boldsymbol{\alpha}_k)
=
\frac{\exp(\boldsymbol{x}^{\top}\boldsymbol{\alpha}_k)}
{\sum_{j=1}^K \exp(\boldsymbol{x}^{\top}\boldsymbol{\alpha}_j)},
\qquad k=1,\ldots,K,
\]
where $\boldsymbol{\alpha}_k$ is a $(p+1)$-dimensional parameter vector.
For identifiability, we impose the constraint $\boldsymbol{\alpha}_K=\boldsymbol{0}$.

\vspace{0.3cm}

From a prediction perspective, the MoE model yields the conditional mean of the response for a new covariate value $\boldsymbol{x}_*$ as
\begin{equation*}
E(Y \mid \boldsymbol{X}=\boldsymbol{x}_*)
=
\sum_{k=1}^K \pi(\boldsymbol{x}_*;\boldsymbol{\alpha}_k)\,
\boldsymbol{x}_*^{\top}\boldsymbol{\beta}_k.
\end{equation*}
This representation highlights the prediction mechanism of the MoE model as a weighted combination of component-specific regression functions.
In this context, the mixing proportions $\pi(\boldsymbol{x};\boldsymbol{\alpha}_k)$, $k=1,2,\ldots, K$, are commonly referred to as the gating network, whereas the component-specific regression functions $\boldsymbol{x}^{\top}\boldsymbol{\beta}_k$ ($k=1,2,\ldots, K$) are referred to as the expert networks.
Accordingly, the MoE model can be interpreted as an ensemble prediction approach in which the contribution of each expert network is adaptively determined by the gating network through the covariates, as illustrated in Figure~\ref{fig:moe_structure}.

\begin{figure}[h]
\centering
\includegraphics[width=0.7\textwidth]{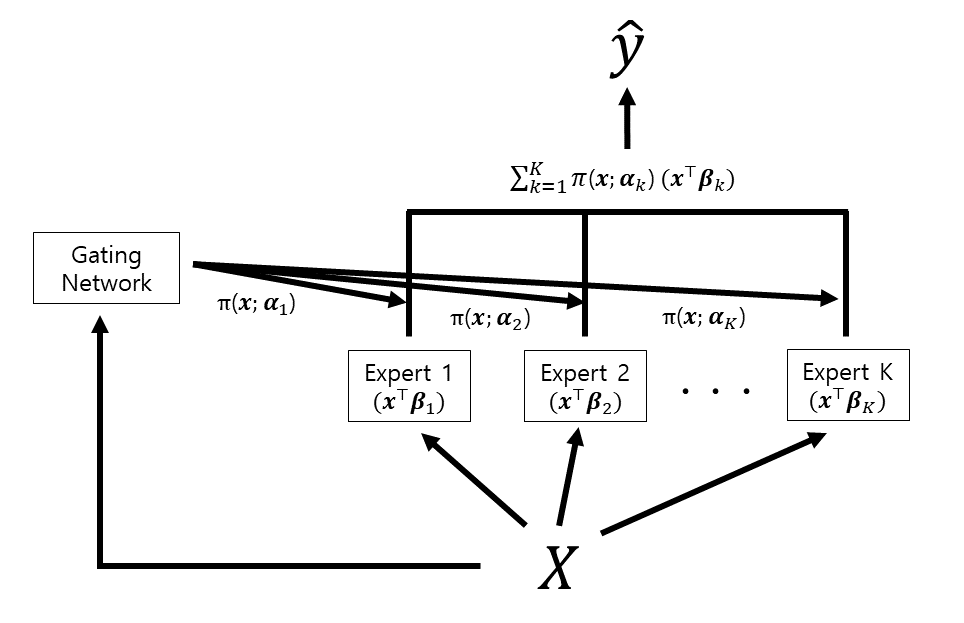}
\caption{Prediction mechanism of the mixture of experts model.}
\label{fig:moe_structure}
\end{figure}

While the classical MoE formulation assumes Gaussian expert errors, this parametric specification may be restrictive in practice when the true error distribution deviates from normality. Under such misspecification, Gaussian expert models can suffer from unstable parameter estimation and degraded predictive performance. Although alternative parametric error distributions have been proposed to improve robustness, they still rely on parametric assumptions that may be insufficient to model complex error structures. This limitation motivates the development of semiparametric expert models that allow the error distribution to adapt flexibly to the data, thereby enhancing robustness without sacrificing efficiency in either normal or nonnormal settings.

\subsection{Gaussian scale mixture distributions}

A Gaussian scale mixture density is defined as
\begin{align}
f(\epsilon;G)
=
\int \frac{1}{\sigma}
\phi\!\left(\frac{\epsilon}{\sigma}\right)
\, dG(\sigma),
\label{NGSM}
\end{align}
where $\phi(\cdot)$ denotes the standard normal density and $G$ is a mixing distribution supported on $\mathbb{R}^+=(0,\infty)$. We write $\epsilon \sim \mathrm{GSM}(0,G)$.
Different choices of $G$ yield a wide class of probability distributions.
For instance, if $G$ is a degenerate distribution with all its mass concentrated at $\sigma_0$, then $f(\epsilon)$ reduces to a normal density with mean zero and variance $\sigma_0^2$.
If $G$ is a discrete distribution with two support points, the resulting density corresponds to a contaminated normal distribution.
More generally, any finite scale mixture of zero-mean normal distributions can be obtained when $G$ is a discrete distribution with finite support.
Moreover, when $\sigma^2$ follows an inverse-gamma distribution, $f(\epsilon)$ corresponds to a Student's $t$ distribution.
It is well known that the class of Gaussian scale mixtures in~\eqref{NGSM} includes many important symmetric unimodal distributions, such as the normal, $t$, Logistic, and Laplace distributions; see, for example, \citet{andrews1974scale}, \citet{efron1978broad}, and \citet{west1987scale}.

Several authors have proposed specifying $G$ using parametric families to construct flexible error distributions.
For example, \citet{garay2016nonlinear}, \citet{garay2017linear}, and \citet{zeller2019finite} considered parametric forms of $G$ to model various unimodal error structures.
Such approaches can be effective when prior knowledge about the underlying error distribution is available.
However, imposing a parametric form on $G$ may lead to model misspecification if the assumed family is incorrect.

To avoid these limitations, one may leave the mixing distribution $G$ unspecified and estimate it nonparametrically based on given data. This semiparametric approach eliminates the need to select a parametric family for $G$ and allows the data to determine the shape of the error distribution flexibly, thereby enhancing robustness to distributional misspecification.

\subsection{Nonparametric maximum likelihood estimator}

In this paper, the mixing distribution $G$ is defined on the positive real line and is left completely unspecified.
Because $G$ is nonparametric, obtaining the MLE for $G$ in \eqref{NGSM} requires optimization over an infinite-dimensional space of probability measures.
The key idea is to reformulate the likelihood maximization problem so that the infinite-dimensional optimization over $G$ is transformed into a finite-dimensional problem defined on a convex set. Classical results in nonparametric mixture theory show that the maximizer can be represented as a discrete distribution with finite support \citep{Lindsay_1995}. Accordingly, the search for the NPMLE may be restricted to probability measures with finite support without loss of generality.

Let $\epsilon_1,\ldots,\epsilon_n$ denote independent realizations from a Gaussian scale mixture density of the form \eqref{NGSM}.
For a given $G$, the likelihood contribution of the $i$th observation can be written as
\[
L_i(G)=\int L_i(\sigma)\,dG(\sigma),
\qquad
L_i(\sigma)=\frac{1}{\sigma}\phi\!\left(\frac{\epsilon_i}{\sigma}\right).
\]
Collecting these terms, the likelihood can be represented by the mixture likelihood vector
\[
\boldsymbol{L}(G) = \bigl(L_1(G),\ldots,L_n(G)\bigr)^\top,
\]
which depends on $G$.
Rather than optimizing directly over the space of all probability measures $G$, the feasible region of the likelihood vectors is characterized as
\[
\Omegab = \left\{ \boldsymbol{L}(G) \mid G \text{ is a probability measure on } \mathbb{R}^+ \right\}.
\]
It can be shown that $\Omega$ is the convex hull of the curve
\[
\Gamma = \left\{ \boldsymbol{L}(\sigma) = \bigl(L_1(\sigma),\ldots,L_n(\sigma)\bigr)^\top : \sigma \in \mathbb{R}^+ \right\}.
\]
The log-likelihood maximization problem can therefore be reformulated as maximizing a concave function over the convex set $\Omega$.
Under mild regularity conditions, the maximizer exists and lies on the boundary of $\Omega$, and the resulting NPMLE $\hat{G}$ is discrete with at most $n$ support points.

To characterize the optimality of the NPMLE, directional derivatives of the log-likelihood functional play a fundamental role.
Let
\[
\ell(G)=\sum_{i=1}^n \log L_i(G),
\]
denote the log-likelihood regarded as a functional of the mixing distribution $G$.
For any candidate scale value $\sigma \in \mathbb{R}^+$, consider a one-dimensional perturbation of $G$ defined by
\[
G_\eta = (1-\eta)G + \eta\,\delta_\sigma, \qquad \eta \in [0,1],
\]
where $\delta_\sigma$ denotes the degenerate distribution placing unit mass at $\sigma$.
The directional derivative of $\ell(G)$ in the direction $\delta_\sigma - G$ is then given by
\[
D(\sigma;G)
= \lim_{\eta \to 0+}
\frac{
{\ell}( G_\eta )
- {\ell}(G)
}{\eta}.
\]

A fundamental result in nonparametric likelihood theory \citep{Lindsay_1995} states that a mixing distribution $\hat{G}$ maximizes the log-likelihood if and only if
\[
D(\sigma;\hat{G}) \le 0 \quad \text{for all } \sigma \in \mathbb{R}^+,
\]
with equality holding at every support point of $\hat{G}$.
This condition is both necessary and sufficient and provides a complete characterization of the NPMLE.
This characterization forms the theoretical foundation for a class of constructive algorithms for computing the NPMLE, including the vertex direction method (VDM; \citealp{Bohning85}), the vertex exchange method (VEM; \citealp{Bohning86}), the intra-simplex direction method (ISDM; \citealp{LK92}), and the constraint Newton method for multiple support points (CNM; \citealp{Wang07}).

\section{Proposed Method} \label{sec:proposed}

\subsection{Semiparametric mixture of experts}

We propose a semiparametric MoE model in which the conditional distribution of $Y$ given $\boldsymbol{x}$ is specified as
\begin{equation}
p(y \mid \boldsymbol{x}; \Thetab, \Gb)
=
\sum_{k=1}^K
\pi(\boldsymbol{x};\boldsymbol{\alpha}_k)
\int
\frac{1}{\sigma}
\phi\!\left(
\frac{y - \boldsymbol{x}^\top \boldsymbol{\beta}_k}{\sigma}
\right)
\, dG_k(\sigma),
\label{eq:spmoe}
\end{equation}
where $\Thetab=\{\boldsymbol{\alpha}_k,\boldsymbol{\beta}_k: k=1,\ldots,K\}$ and $\bm{G} = (G_1, G_2, \ldots, G_K)$.
Model~\eqref{eq:spmoe} generalizes classical MoE models by representing each component error distribution as a Gaussian scale mixture with an unspecified mixing distribution $G_k$. This semiparametric specification relaxes fixed parametric error assumptions, enabling component error distributions to be estimated directly from the data while preserving the gating–expert structure of the MoE framework and enhancing robustness to heavy-tailed errors and outliers.

\citet{Teicher1963} established identifiability for finite mixture models using necessary and sufficient conditions. \cite{hennig2000identifiability} showed that finite mixtures of regressions with Gaussian
components are identifiable when the mixing proportions do not depend on covariates
and the support of $\boldsymbol{U}$ contains an open set in $\mathbb{R}^p$.
Furthermore, \citet{oh2024semiparametric} established the identifiability of Gaussian scale mixture error distributions based on the NPMLE when the mixing proportions do not depend on covariates. \citet{XYS16} showed identifiability for semiparametric mixture models with nonparametric Gaussian scale mixture components, but their result is restricted to mixtures with a limited number of components ($K \le 3$) and a common mixing distribution ($G_1 = \cdots = G_K$).

Motivated by these limitations, we first establish an identifiability result for finite mixtures whose component densities are Gaussian scale mixtures with unspecified mixing distributions. The result accommodates an arbitrary number of components and component-specific nonparametric mixing distributions $G_1, G_2,\ldots, G_K$. To the best of our knowledge, existing identifiability results for semiparametric mixtures do not cover this level of generality, as they either restrict the number of components or require a common mixing distribution across components. Lemma~1 therefore provides the key theoretical bridge from classical finite-mixture identifiability to the proposed semiparametric MoE model with distinct expert-specific error distributions.

The following assumptions specify regularity conditions for the covariate distribution, the expert regression parameters, and the class of scale mixing distributions. Assumptions~\textup{(A1)}--\textup{(A2)} ensure sufficient
variation in the covariates and separation of the expert regression functions.
Assumption~\textup{(A3)} imposes an asymptotic separation condition on the scale mixing distributions through the decay rates of the corresponding characteristic functions. Assumptions~\textup{(A4)}--\textup{(A5)} provide regularity of the scale mixing distribution class and compactness of the finite-dimensional parameter space, which are used in the consistency argument.

Let \(\mathcal G\) be a class of scale mixing distributions. For
\(G\in\mathcal G\), let \(\Psi_G\) denote the characteristic function of the
centered Gaussian scale mixture distribution induced by \(G\).

\medskip
\noindent\textbf{Assumptions.}
\begin{enumerate}
\item[(A1)] The support of $\boldsymbol{U}$ contains a nonempty open set in $\mathbb{R}^p$.

\item[(A2)] The regression coefficient vectors
\(\boldsymbol\beta_1,\ldots,\boldsymbol\beta_K\) are pairwise distinct.

\item[(A3)] For any distinct \(G,G'\in\mathcal G\), 
\[
\frac{\Psi_{G'}(t)}{\Psi_G(t)} \to 0 \textrm{ or } \infty
\qquad\text{as } |t|\to\infty .
\]

\item[(A4)] There exist constants $\ell>0$ and $M<\infty$ such that,
for every $G\in\mathcal G$,
\[
\operatorname{supp}(G)\subset[\ell,\infty)
\qquad\text{and}\qquad
\int_{\ell}^{\infty}\log\sigma\,dG(\sigma)\le M.
\]

\item[(A5)] The parameter space of
\[
\Theta=\{\boldsymbol\alpha_k,\boldsymbol\beta_k:k=1,\ldots,K\}
\]
is compact.
\end{enumerate}
Under these conditions, the following lemma formalizes identifiability for finite mixture models with distinct nonparametric Gaussian scale mixture components.

\begin{lemma}
\label{lem:gsm-ident}
Suppose that Assumption~\textup{(A3)} holds. Let
\(G_1,\ldots,G_K,\tilde G_1,\ldots,\tilde G_{\tilde K}\in\mathcal G\).
Assume that \(\mu_1,\ldots,\mu_K\) are pairwise distinct and that
\(\tilde\mu_1,\ldots,\tilde\mu_{\tilde K}\) are pairwise distinct.
Let \(\pi_1,\ldots,\pi_K\) and
\(\tilde\pi_1,\ldots,\tilde\pi_{\tilde K}\) satisfy
\[
\pi_k>0,\quad \sum_{k=1}^K \pi_k=1,
\qquad
\tilde\pi_j>0,\quad \sum_{j=1}^{\tilde K}\tilde\pi_j=1.
\]
If, for all \(x\in\mathbb R\),
\[
\sum_{k=1}^{K} \pi_k 
\int \frac{1}{\sigma}
\phi\!\left(\frac{x-\mu_k}{\sigma}\right)\, dG_k(\sigma)
=
\sum_{j=1}^{\tilde K} \tilde\pi_j 
\int \frac{1}{\sigma}
\phi\!\left(\frac{x-\tilde\mu_j}{\sigma}\right)\, d\tilde G_j(\sigma),
\]
then \(\tilde K = K\) and there exists a permutation \(\tau\) of
\(\{1,\ldots,K\}\) such that, for all \(k\),
\[
\pi_k = \tilde\pi_{\tau(k)},\qquad
\mu_k = \tilde\mu_{\tau(k)},\qquad
G_k = \tilde G_{\tau(k)}.
\]
\end{lemma}

\begin{proof}
A proof is given in Appendix~A.
\end{proof}

Lemma~\ref{lem:gsm-ident} provides the identifiability of finite mixtures of Gaussian scale mixtures with distinct locations. Building on this result, we establish the identifiability of \eqref{eq:spmoe} in the following theorem.

\begin{theorem}
\label{thm:ident}
Suppose that Assumptions~\textup{(A1)}--\textup{(A3)} hold. If two sets
\(\{\alphab_k,\betab_k,G_k\}_{k=1}^K\) and
\(\{\tilde\alphab_j,\tilde\betab_j,\tilde G_j\}_{j=1}^{\tilde K}\) yield the
same conditional density \eqref{eq:spmoe} for all
\((y,\ub)\in\mathbb R\times\mathbb R^p\), then $\tilde K=K$ and there exists a permutation
$\tau$ of $\{1,\ldots,K\}$ such that
\[
\betab_k=\tilde\betab_{\tau(k)},\qquad
G_k=\tilde G_{\tau(k)},\qquad
\alphab_k
=
\tilde\alphab_{\tau(k)}-\tilde\alphab_{\tau(K)},
\qquad k=1,\ldots,K.
\]
\end{theorem}

\begin{proof}
A proof is given in Appendix~B.
\end{proof}

The consistency of structural parameters and latent mixing distributions in semiparametric mixture models was established by \citet{kiefer1956consistency}. Since the consistency of the maximum likelihood estimator requires identifiability of
the underlying model, Theorem~\ref{thm:ident} provides the required identifiability condition for our setting, where the number of experts $K$ is known and fixed.

\begin{theorem}
\label{thm:consistency}
Let $(\Thetab_0,\Gb_0)$ denote the true parameters and mixing
distributions, and let $h_0$ denote the marginal density of $\Ub$.
Suppose that Assumptions~\textup{(A1)}--\textup{(A5)} hold and that
the number of experts $K$ is known and fixed. Assume further that
\[
\|h_0\|_{\infty}<\infty,
\qquad
\mathbb E_0[-\log h_0(\Ub)]<\infty,
\qquad
\mathbb E_0\|\Ub\|^2<\infty,
\qquad
\mathbb E_0Y^2<\infty.
\]
Then, under model~\eqref{eq:spmoe}, any global maximizer
$(\widehat{\Thetab}_n,\widehat{\Gb}_n)$ of the log-likelihood is
consistent for $(\Thetab_0,\Gb_0)$, up to a permutation of the
component labels.
\end{theorem}

\begin{proof}
A proof is given in Appendix~C.
\end{proof}

\subsection{Expectation--conditional maximization algorithm} \label{sec:em}

Parameter estimation is carried out via an expectation-conditional maximization (ECM) algorithm \citep{meng1993maximum}. The observed log-likelihood is given by
\begin{align*}
\ell(\Thetab, \Gb)
&=
\sum_{i=1}^n
\log
\left\{
\sum_{k=1}^K
\pi(\boldsymbol{x}_i;\boldsymbol{\alpha}_k)
\int
\frac{1}{\sigma}
\phi\!\left(
\frac{y_i - \boldsymbol{x}_i^\top \boldsymbol{\beta}_k}{\sigma}
\right)
\, dG_k(\sigma)
\right\}.
\end{align*}
Let $Z_{ik}=\mathbb{I}(Z_i=k)$ denote the latent membership indicator.
The complete log-likelihood can be written as
\begin{align*}
\ell_c(\Thetab, \Gb)
&=
\sum_{i=1}^n \sum_{k=1}^K z_{ik}
\Bigg[
\log \pi(\boldsymbol{x}_i;\boldsymbol{\alpha}_k)
+
\log
\int
\frac{1}{\sigma}
\phi\!\left(
\frac{y_i - \boldsymbol{x}_i^\top \boldsymbol{\beta}_k}{\sigma}
\right)
\, dG_k(\sigma)
\Bigg].
\end{align*}

\paragraph{E-step.}
At iteration $t$, the posterior probabilities are computed as
\begin{equation*}
z_{ik}^{(t)}
=
\frac{
\pi(\boldsymbol{x}_i;\boldsymbol{\alpha}_k^{(t)})
\int \frac{1}{\sigma}
\phi\!\left(
\frac{y_i - \boldsymbol{x}_i^\top \boldsymbol{\beta}_k^{(t)}}{\sigma}
\right)
\, dG_k^{(t)}(\sigma)
}{
\sum_{\ell=1}^K
\pi(\boldsymbol{x}_i;\boldsymbol{\alpha}_\ell^{(t)})
\int \frac{1}{\sigma}
\phi\!\left(
\frac{y_i - \boldsymbol{x}_i^\top \boldsymbol{\beta}_\ell^{(t)}}{\sigma}
\right)
\, dG_\ell^{(t)}(\sigma)
}.
\end{equation*}

\paragraph{CM-step 1: Updating $\Gb^{(t+1)}$.}
For each component $k$, the mixing distribution $G_k$ is updated by NPMLE, conditional on the current values of the remaining parameters.
For a candidate scale value $\sigma \in \mathbb{R}^+$, the directional derivative is given by
\begin{align*}
D_{G_k}(\sigma)
&=
\sum_{i=1}^{n} z_{ik}^{(t)}
\frac{
\frac{1}{\sigma}
\phi\!\left(
\frac{y_i - \boldsymbol{x}_i^\top \boldsymbol{\beta}_k^{(t)}}{\sigma}
\right)
}{
\int \frac{1}{\sigma}
\phi\!\left(
\frac{y_i - \boldsymbol{x}_i^\top \boldsymbol{\beta}_k^{(t)}}{\sigma}
\right)
\, dG_k(\sigma)
}
-
\sum_{i=1}^{n} z_{ik}^{(t)}.
\end{align*}
The updated mixing distribution $G_k^{(t+1)}$ is obtained by iteratively modifying the support and weights of $G_k$ using the constraint Newton method for multiple support points (CNM), which exploits its directional-derivative characterization.

\paragraph{CM-step 2: Updating $\boldsymbol{\alpha}^{(t+1)}$.}
The gating parameters $\boldsymbol{\alpha}
= (\boldsymbol{\alpha}_1, \boldsymbol{\alpha}_2, \ldots, \boldsymbol{\alpha}_K)^{\top}$
are updated by solving
\begin{equation}
\label{alpha}
\boldsymbol{\alpha}^{(t+1)}
=
\arg\max_{\boldsymbol{\alpha}}
\sum_{i=1}^n \sum_{k=1}^K
z_{ik}^{(t)}
\left[
\boldsymbol{x}_i^\top \boldsymbol{\alpha}_k
-
\log \sum_{\ell=1}^K \exp(\boldsymbol{x}_i^\top \boldsymbol{\alpha}_\ell)
\right],
\end{equation}
subject to the identifiability constraint $\boldsymbol{\alpha}_K=\mathbf{0}$,
which corresponds to a multinomial logistic regression problem.
This optimization can be efficiently carried out using either the iteratively reweighted least squares (IRLS) algorithm or a quasi-Newton method such as the Broyden--Fletcher--Goldfarb--Shanno (BFGS) algorithm
(\citealp{broyden1970convergence}; \citealp{fletcher1970new};
\citealp{goldfarb1970family}; \citealp{shanno1970conditioning}).

\paragraph{CM-step 3: Updating $\boldsymbol{\beta}^{(t+1)}$.}
For each component $k$, the regression coefficients are updated by maximizing
\begin{equation*}
\sum_{i=1}^n
z_{ik}^{(t)}
\log
\int
\frac{1}{\sigma}
\phi\!\left(
\frac{y_i - \boldsymbol{x}_i^\top \boldsymbol{\beta}_k}{\sigma}
\right)
\, dG_k^{(t+1)}(\sigma).
\end{equation*}
Since $G_k^{(t+1)}$ is discrete with support points
$\{(\hat{\sigma}_{ks},\hat{\omega}_{ks}): s=1,\ldots,S_k\}$,
an inner EM algorithm can be employed.
Specifically, introduce latent indicators $w_{iks}$ indicating whether the $i$th observation is associated with the $s$th scale component within the $k$th mixture.
The conditional expectations are given by
\begin{equation*}
\tau_{iks}^{(t)}
=
\frac{
\hat{\omega}_{ks}
\frac{1}{\hat{\sigma}_{ks}}
\phi\!\left(
\frac{y_i - \boldsymbol{x}_i^\top \boldsymbol{\beta}_k^{(t)}}{\hat{\sigma}_{ks}}
\right)
}{
\sum_{r=1}^{S_k}
\hat{\omega}_{kr}
\frac{1}{\hat{\sigma}_{kr}}
\phi\!\left(
\frac{y_i - \boldsymbol{x}_i^\top \boldsymbol{\beta}_k^{(t)}}{\hat{\sigma}_{kr}}
\right)
}.
\end{equation*}
Then, $\boldsymbol{\beta}_k^{(t+1)}$ is obtained via weighted least squares,
\begin{equation*}
\boldsymbol{\beta}_k^{(t+1)}
=
(\boldsymbol{X}^\top \boldsymbol{V}_k^{(t)} \boldsymbol{X})^{-1}
\boldsymbol{X}^\top \boldsymbol{V}_k^{(t)} \boldsymbol{y},
\end{equation*}
where $\boldsymbol{X}$ is the $n \times (p+1)$ design matrix and
$\boldsymbol{V}_k^{(t)}$ is a diagonal matrix with diagonal elements
$z_{ik}^{(t)} \sum_{s=1}^{S_k} \tau_{iks}^{(t)} / \hat{\sigma}_{ks}^2$.

Proposition~1 establishes the monotone increasing property of the ECM algorithm. In particular, the algorithm produces a non-decreasing sequence of observed log-likelihood values. Iterations are continued until convergence, defined by the condition
\[
\ell(\Thetab^{(t+1)}, \Gb^{(t+1)}) - \ell(\Thetab^{(t)}, \Gb^{(t)}) < \delta,
\]
for a sufficiently small tolerance $\delta > 0$, or until a pre-specified maximum number of iterations is reached.

\begin{proposition}
The ECM algorithm generates a non-decreasing sequence of observed
log-likelihood values, that is,
\[
\ell(\Thetab^{(t+1)},\Gb^{(t+1)})
\;\ge\;
\ell(\Thetab^{(t)},\Gb^{(t)}),
\qquad t=0,1,2,\ldots
\]
\end{proposition}

\begin{proof}
Let $Q(\Thetab,\Gb \mid \Thetab^{(t)},\Gb^{(t)})$ denote the conditional expectation
of the complete log-likelihood given the observed data.
The standard EM inequality yields
\[
\ell(\Thetab^{(t+1)},\Gb^{(t+1)}) - \ell(\Thetab^{(t)},\Gb^{(t)})
\;\ge\;
Q(\Thetab^{(t+1)},\Gb^{(t+1)} \mid \Thetab^{(t)},\Gb^{(t)})
-
Q(\Thetab^{(t)},\Gb^{(t)} \mid \Thetab^{(t)},\Gb^{(t)}).
\]
Each CM-step is constructed so as to increase (or at least not decrease) the
corresponding block of the $Q$-function.
In particular, in CM-step~1, $\Gb^{(t+1)}$ is taken as the NPMLE, which maximizes
$Q$ over the space of probability measures on $\mathbb{R}^+$.
Hence, the right-hand side is nonnegative, which implies the desired monotonicity.
\end{proof}

\section{Simulation Study} \label{sec:sim}

We conduct a Monte Carlo simulation study to evaluate the finite-sample behavior of the proposed semiparametric MoE model under a range of error distributions. The proposed method is compared with MoE models under normal and $t$-distributed errors. For convenience, we refer to these competing methods as the Normal and $t$ models, respectively, and denote the proposed semiparametric estimator as NPMLE.

For all simulation settings, data are generated from a two-component latent structure with covariates $\mathbf{X} = (1, X_1, X_2)^\top$, where
\[
X_1, X_2 \overset{\text{i.i.d.}}{\sim} N(0, 10^2).
\]
The latent class indicator $Z \in \{1,2\}$ follows a logistic gating mechanism
\[
P(Z=1 \mid \mathbf{X}) = \pi_1(\mathbf{X}), 
\qquad 
\pi_1(\mathbf{X}) 
= \frac{\exp(\mathbf{X}^\top \boldsymbol{\alpha}_1)}
{1+\exp(\mathbf{X}^\top \boldsymbol{\alpha}_1)},
\]
with $\boldsymbol{\alpha}_1=(0,1,1)^\top$. Conditional on $Z$, the response variable is generated as
\[
Y =
\begin{cases}
\mathbf{X}^\top \boldsymbol{\beta}_1 + \varepsilon_1, & Z=1,\\
\mathbf{X}^\top \boldsymbol{\beta}_2 + \varepsilon_2, & Z=2,
\end{cases}
\]
where
\[
\boldsymbol{\beta}_1=(0,-3,3)^\top,
\qquad
\boldsymbol{\beta}_2=(0,3,-3)^\top.
\]

The component errors $(\varepsilon_1,\varepsilon_2)$ are generated under six cases. Throughout this section, $\mathcal{N}(\mu,\sigma^2)$ denotes the normal distribution with mean $\mu$ and variance $\sigma^2$, $t_\nu$ denotes the central Student's $t$ distribution with $\nu$ degrees of freedom, and $t_\nu(\lambda)$ denotes the noncentral Student's $t$ distribution with $\nu$ degrees of freedom and noncentrality parameter $\lambda$. The notation $\mathrm{Laplace}(0,b)$ denotes the Laplace distribution with location parameter 0 and scale parameter $b$.

\begin{itemize}
\item Case $\uppercase\expandafter{\romannumeral1}$: $\varepsilon_1,\varepsilon_2 \sim \mathcal{N}(0,1)$
\item Case $\uppercase\expandafter{\romannumeral2}$: $\varepsilon_1,\varepsilon_2 \sim t_3$
\item Case $\uppercase\expandafter{\romannumeral3}$: $\varepsilon_1,\varepsilon_2 \sim \mathrm{Laplace}(0,1)$
\item Case $\uppercase\expandafter{\romannumeral4}$: $\varepsilon_1,\varepsilon_2 \sim 0.8\,\mathcal{N}(0,1) + 0.2\,\mathcal{N}(0,5^2)$
\item Case $\uppercase\expandafter{\romannumeral5}$: $\varepsilon_1 \sim 0.8\,\mathcal{N}(0,1) + 0.2\,\mathcal{N}(0,5^2)$ and $\varepsilon_2 \sim 0.5\,\mathcal{N}(0,1) + 0.3\,\mathcal{N}(0,5^2) + 0.2\,\mathcal{N}(0,9^2)$
\item Case $\uppercase\expandafter{\romannumeral6}$: $\varepsilon_1 \sim 0.5\,t_3(-1) + 0.5\,t_3(1)$ and 
$\varepsilon_2 \sim 0.5\,\mathcal{N}(0,1) + 0.3\,\mathcal{N}(0,5^2) + 0.2\,\mathcal{N}(0,9^2)$
\end{itemize}
Case $\uppercase\expandafter{\romannumeral1}$ corresponds to a well-specified Gaussian setting without contamination. 
The remaining cases introduce heavy tails and mixture contamination to assess robustness and efficiency of each method under model misspecification and the presence of outliers. 

\subsection{Simulation 1: Estimation accuracy with known $K$}

In the first experiment, the number of latent components is assumed to be known and fixed at $K=2$, thereby isolating estimation accuracy from model selection effects. Independent samples of sizes $n=200$ and $n=400$ are generated under each error case described in Section~\ref{sec:sim}. 
Let $\thetab=(\alphab_1^\top,\betab_1^\top,\betab_2^\top)^\top$ denote the vector of finite-dimensional parameters, and let $\thetab_0$ be the true parameter vector used to generate the data. 

For the $r$th Monte Carlo replication, estimation accuracy is quantified by 
\[
\gammab_r 
=
\bigl\|\hat{\thetab}_r-\thetab_0\bigr\|_2,
\]
which measures the overall deviation of the estimator from the true parameter.
Across $R=200$ replications, performance is summarized by the average estimation error
\[
\bar{\gammab}
=
\frac{1}{R}\sum_{r=1}^{R}\gammab_r,
\]
together with its sample standard deviation
\[
s(\gammab)
=
\sqrt{
\frac{1}{R-1}
\sum_{r=1}^{R}
(\gammab_r-\bar{\gammab})^2
}.
\]
These quantities reflect both the accuracy and the stability of the estimators.

Table~\ref{tab:sim1} reports the mean estimation error and its variability across replications. 
Under the Gaussian setting (Case I), where the parametric normal model is correctly specified, the Normal estimator achieves the smallest estimation error for both sample sizes.  All methods exhibit similar performance in this setting, indicating that the flexibility of the NPMLE approach does not introduce noticeable inefficiency when the error distribution is Gaussian. As the sample size increases from $n=200$ to $n=400$, estimation errors decrease substantially for all methods, consistent with expected asymptotic behavior.

In contrast, under heavy-tailed or contaminated settings (Cases II–VI), clear performance differences arise. The Normal estimator becomes increasingly sensitive to model misspecification, exhibiting inflated estimation errors and greater variability, particularly for smaller sample sizes. The $t$-based estimator improves robustness in moderately heavy-tailed scenarios, but its performance deteriorates when the error structure involves more complex mixture contamination.
Across all cases, the proposed NPMLE-based semiparametric estimator consistently attains the smallest or near-smallest estimation error. Its advantage is most pronounced in non-Gaussian settings (Cases III–VI), where it maintains both lower estimation error and reduced variability compared to the parametric alternatives. This behavior demonstrates the adaptive nature of the semiparametric approach, which effectively balances robustness and efficiency across diverse error structures.

\begin{table}[ht]
\centering
\caption{Mean estimation error $\bar{\gammab}$ and corresponding sample standard deviation $s(\gammab)$ (in parentheses) computed from $R=200$ replications (Boldfaced numbers indicate the smallest value in each criterion).}

\label{tab:sim1}

\begin{tabular}{c | ccc | ccc}
\hline\hline \noalign{\smallskip}
& \multicolumn{3}{c}{$n=200$}
& \multicolumn{3}{c}{$n=400$} \\
\cline{2-4} \cline{5-7}
Case
& Normal & $t$ & NPMLE
& Normal & $t$ & NPMLE \\
\hline

$\uppercase\expandafter{\romannumeral1}$   & $\bm{0.748}$ (0.80) & 0.750 (0.80) & 0.752 (0.80) & 
$\bm{0.449}$ (0.31) & 0.451 (0.31) & 0.452 (0.31) \\
$\uppercase\expandafter{\romannumeral2}$  & 1.348 (5.17) & $\bm{0.784}$ (0.61) & 0.789 (0.61) & 
0.495 (0.23) & $\bm{0.453}$ (0.22) & 0.459 (0.21) \\
$\uppercase\expandafter{\romannumeral3}$ & 1.048 (3.16) & 0.709 (0.78) & $\bm{0.707}$ (0.78) &
0.459 (0.25) & 0.446 (0.26) & $\bm{0.443}$ (0.26) \\
$\uppercase\expandafter{\romannumeral4}$  & 1.627 (7.21) & 1.310 (4.02) & $\bm{1.037}$ (1.40) & 
0.609 (0.31) & 0.518 (0.33) & $\bm{0.512}$ (0.33) \\
$\uppercase\expandafter{\romannumeral5}$  & 1.314 (1.21) & 1.083 (1.33) & $\bm{1.067}$ (1.32) & 
0.823 (0.39) & 0.555 (0.33) & $\bm{0.547}$ (0.32) \\
$\uppercase\expandafter{\romannumeral6}$   & 1.309 (1.31) & 1.068 (1.36) & $\bm{1.037}$ (1.32) &
0.814 (0.57) & 0.565 (0.50) & $\bm{0.564}$ (0.50) \\
\hline
\end{tabular}
\end{table}

\subsection{Simulation 2: Clustering performance with unknown $K$}

In this experiment, the number of latent components is treated as unknown in order to evaluate clustering performance under model-selection uncertainty. Samples of size $n=200$ are generated under the six cases described in Section~\ref{sec:sim}. For each simulated data set, competing models are fitted with $K=1,2,3$, and the final model is selected using either the Bayesian information criterion (BIC; \citealp{schwarz1978estimating}) or the integrated completed likelihood (ICL; \citealp{biernacki2000assessing}).

For each replication, clustering performance is evaluated by comparing estimated cluster assignments with the true latent labels using the adjusted Rand index (ARI) and adjusted mutual information (AMI), both of which are invariant to label permutations. Reported ARI and AMI values are averages over $R=200$ replications.
In addition to clustering accuracy, we record the proportion with which each method and information criterion selects $K=1$, $2$, or $3$. These proportions provide insight into the stability of model selection, while ARI and AMI quantify the quality of cluster recovery.

Table~\ref{tab:sim2} summarizes model-selection frequencies together with average clustering accuracy. Under the Gaussian setting (Case~\uppercase\expandafter{\romannumeral1}), all methods correctly select $K=2$ with high probability, and clustering accuracy is uniformly high. Differences among estimators are minor in this well-specified setting, indicating that all approaches recover the latent structure reliably when the error model is correctly specified.
When the error distribution deviates from normality (Cases~\uppercase\expandafter{\romannumeral2}--\uppercase\expandafter{\romannumeral6}), the Normal model shows reduced stability in selecting the correct number of components, with a noticeable tendency to overestimate $K$, particularly under contamination. This instability is reflected in lower ARI and AMI values relative to the competing methods.

\begin{table}[p]
\centering
\caption{Proportions of selected cluster numbers over $R=200$ replications along with ARI and AMI (Boldface indicates the best performance in each setting).}
\label{tab:sim2}

\begin{tabular}{c c c c c c c c}

\hline\hline \noalign{\smallskip}

Case & Method & Criterion 
& $K=1$ & $K=2$ & $K=3$ 
& ARI & AMI \\
\hline

\multirow{6}{*}{$\uppercase\expandafter{\romannumeral1}$}
& \multirow{2}{*}{Normal}
& BIC & 4.5 $\%$& 93.5 $\%$ & 2.0 $\%$ & 0.9426 & 0.9372 \\
& & ICL & 4.5 $\%$ & 93.5 $\%$ & 2.0 $\%$ & 0.9426 & 0.9372 \\

& \multirow{2}{*}{$t$}
& BIC & 0.5 $\%$ & 99.5 $\%$ & 0.0 $\%$ & 0.9814 & 0.9779 \\
& & ICL & 0.5 $\%$ & 99.5 $\%$ & 0.0 $\%$ & 0.9814 & 0.9779 \\

& \multirow{2}{*}{NPMLE}
& BIC & 0.0 $\%$ & $\bm{100.0}$ $\%$ & 0.0 $\%$ & $\bm{0.9919}$ & $\bm{0.9886}$ \\
& & ICL & 0.0 $\%$ & $\bm{100.0}$ $\%$ & 0.0 $\%$ & $\bm{0.9919}$ & $\bm{0.9886}$ \\

\hline

\multirow{6}{*}{$\uppercase\expandafter{\romannumeral2}$}
& \multirow{2}{*}{Normal}
& BIC & 2.5 $\%$ & 90.0 $\%$ & 7.5 $\%$ & 0.9553 & 0.9418 \\
& & ICL & 2.5 $\%$ & 94.0 $\%$ & 3.5 $\%$ & 0.9555 & 0.9421 \\

& \multirow{2}{*}{$t$}
& BIC & 0.0 $\%$ & 99.0 $\%$ & 1.0 $\%$ & 0.9934 & 0.9869 \\
& & ICL & 0.0 $\%$ & 99.0 $\%$ & 1.0 $\%$ & $\bm{0.9934}$ & $\bm{0.9870}$ \\

& \multirow{2}{*}{NPMLE}
& BIC & 0.0 $\%$ & $\bm{100.0}$ $\%$ & 0.0 $\%$ & 0.9895 & 0.9839 \\
& & ICL & 0.0 $\%$ & $\bm{100.0}$ $\%$ & 0.0 $\%$ & 0.9896 & 0.9840 \\

\hline

\multirow{6}{*}{$\uppercase\expandafter{\romannumeral3}$}
& \multirow{2}{*}{Normal}
& BIC & 2.0 $\%$ & 95.0 $\%$ & 3.0 $\%$ & 0.9703 & 0.9639 \\
& & ICL & 2.0 $\%$ & 95.0 $\%$ & 3.0 $\%$ & 0.9705 & 0.9641 \\

& \multirow{2}{*}{$t$}
& BIC & 0.0 $\%$ & $\bm{100.0}$ $\%$ & 0.0 $\%$ & 0.9915 & 0.9879 \\
& & ICL & 0.0 $\%$ & $\bm{100.0}$ $\%$ & 0.0 $\%$ & $\bm{0.9915}$ & $\bm{0.9880}$ \\

& \multirow{2}{*}{NPMLE}
& BIC & 0.0 $\%$ & 99.5 $\%$ & 0.5 $\%$ & 0.9901 & 0.9858 \\
& & ICL & 0.0 $\%$ & $\bm{100.0}$ $\%$ & 0.0 $\%$ & 0.9901 & 0.9859 \\

\hline

\multirow{6}{*}{$\uppercase\expandafter{\romannumeral4}$}
& \multirow{2}{*}{Normal}
& BIC & 4.5 $\%$ & 62.0 $\%$ & 33.5 $\%$ & 0.9018 & 0.8553 \\
& & ICL & 4.5 $\%$ & 72.5 $\%$ & 23.0 $\%$ & 0.9019 & 0.8550 \\

& \multirow{2}{*}{$t$}
& BIC & 1.5 $\%$ & 95.0 $\%$ & 3.5 $\%$ & 0.9661 & 0.9538 \\
& & ICL & 1.5 $\%$ & 95.0 $\%$ & 3.5 $\%$ & 0.9663 & 0.9541 \\

& \multirow{2}{*}{NPMLE}
& BIC & 1.0 $\%$ & 97.5 $\%$ & 1.5 $\%$ & 0.9746 & 0.9641 \\
& & ICL & 1.0 $\%$ & $\bm{98.5}$ $\%$ & 0.5 $\%$ & $\bm{0.9747}$ & $\bm{0.9643}$ \\

\hline

\multirow{6}{*}{$\uppercase\expandafter{\romannumeral5}$}
& \multirow{2}{*}{Normal}
& BIC & 0.0 $\%$ & 54.0 $\%$ & 46.0 $\%$ & 0.9043 & 0.8445 \\
& & ICL & 0.0 $\%$ & 71.0 $\%$ & 29.0 $\%$ & 0.9047 & 0.8451 \\

& \multirow{2}{*}{$t$}
& BIC & 1.0 $\%$ & 96.5 $\%$ & 2.5 $\%$ & 0.9729 & 0.9580 \\
& & ICL & 1.0 $\%$ & 96.5 $\%$ & 2.5 $\%$ & 0.9730 & 0.9582 \\

& \multirow{2}{*}{NPMLE}
& BIC & 0.0 $\%$ & 97.0 $\%$ & 3.0 $\%$ & 0.9801 & 0.9653 \\
& & ICL & 0.0 $\%$ & $\bm{99.5}$ $\%$ & 0.5 $\%$ & $\bm{0.9802}$ & $\bm{0.9654}$ \\

\hline

\multirow{6}{*}{$\uppercase\expandafter{\romannumeral6}$}
& \multirow{2}{*}{Normal}
& BIC & 3.0 $\%$ & 66.5 $\%$ & 30.5 $\%$ & 0.9002 & 0.8684 \\
& & ICL & 3.0 $\%$ & 83.5 $\%$ & 13.5 $\%$ & 0.9007 & 0.8690 \\

& \multirow{2}{*}{$t$}
& BIC & 11.5 $\%$ & 85.5 $\%$ & 3.0 $\%$ & 0.8757 & 0.8695 \\
& & ICL & 11.5 $\%$ & 85.5 $\%$ & 3.0 $\%$ & 0.8763 & 0.8701 \\

& \multirow{2}{*}{NPMLE}
& BIC & 0.5 $\%$ & 94.5 $\%$ & 5.0 $\%$ & 0.9812 & 0.9721 \\
& & ICL & 0.5 $\%$ & $\bm{99.0}$ $\%$ & 0.5 $\%$ & $\bm{0.9813}$ & $\bm{0.9722}$ \\

\hline

\end{tabular}
\end{table}

In contrast, both the $t$-based and NPMLE estimators exhibit substantially higher model-selection stability than the Normal estimator, consistently selecting $K=2$ with high frequency across most non-Gaussian cases. Their clustering accuracy remains strong overall, with ARI and AMI values close to one in many settings. The performance of the $t$ and NPMLE estimators is broadly comparable in Cases II--V, although the NPMLE approach often achieves slightly higher selection stability and clustering accuracy in contaminated cases. The advantage of NPMLE becomes particularly pronounced in Case VI, where the error distribution combines a noncentral $t$ mixture and normal-mixture contamination. In this more complex setting, the NPMLE estimator selects the true two-component structure more frequently and yields substantially higher ARI and AMI values than the $t$-based estimator.

These results indicate that robust error modeling improves recovery of the latent cluster structure when the number of components is unknown. In particular, the NPMLE approach provides additional protection when the error distribution departs from standard heavy-tailed parametric forms, supporting reliable model selection and accurate clustering across diverse error structures.

\subsection{Simulation 3: Predictive performance with unknown $K$}

This experiment evaluates predictive performance when the number of mixture components is unknown. For each Monte Carlo replication, the training sample is generated under one of the error cases described in Section~\ref{sec:sim}, while an independent test sample of size $n_{\text{test}}=2000$ is generated from the Gaussian setting (Case~\uppercase\expandafter{\romannumeral1}). This design allows the training data to contain heavy tails or contamination, whereas prediction accuracy is assessed under a clean test distribution. Models are fitted using the training data with $K$ selected via BIC or ICL, and predictions are evaluated on the test sample.
For the $r$th replication, predictive accuracy is measured using the root mean squared error (RMSE) and the mean absolute error (MAE),
\[
\text{RMSE}_{\text{test}}^{(r)}
=
\sqrt{
\frac{1}{n_{\text{test}}}
\sum_{i=1}^{n_{\text{test}}}
\bigl(
y_i^{(r)}-\hat{y}_i^{(r)}
\bigr)^2
},
\qquad
\text{MAE}_{\text{test}}^{(r)}
=
\frac{1}{n_{\text{test}}}
\sum_{i=1}^{n_{\text{test}}}
\left|
y_i^{(r)}-\hat{y}_i^{(r)}
\right|.
\]
Reported values are averages over $R=200$ replications.

Table~\ref{tab:sim3} compares predictive accuracy across competing methods under each case and model-selection criterion. In the Gaussian setting (Case~\uppercase\expandafter{\romannumeral1}), all estimators achieve comparable prediction errors, with the NPMLE method yielding the smallest RMSE and MAE. This indicates that the semiparametric approach remains competitive even when the training distribution is well specified.
When the training data are generated from heavy-tailed or contaminated scenarios (Cases~\uppercase\expandafter{\romannumeral2}--\uppercase\expandafter{\romannumeral6}), clearer differences emerge. The Normal estimator generally exhibits larger prediction errors, reflecting its sensitivity to misspecified training distributions. The $t$-based estimator improves predictive robustness under moderately heavy-tailed settings, particularly in Cases~\uppercase\expandafter{\romannumeral2} and \uppercase\expandafter{\romannumeral3}. However, its performance deteriorates when the error distribution is more complex than a standard heavy-tailed parametric form. This is most evident in Case~\uppercase\expandafter{\romannumeral6}, where the $t$-based estimator yields substantially larger RMSE and MAE than both the Normal and NPMLE estimators.

\begin{table}[ht]
\caption{Average RMSE and MAE are computed from $R=200$ replications
(Boldface indicates the best performance in each setting).}
\label{tab:sim3}
\centering

\begin{tabular}{c c cc cc cc}
\hline\hline \noalign{\smallskip}

\multirow{2}{*}{Case} &
\multirow{2}{*}{Criterion} &
\multicolumn{2}{c}{Normal} &
\multicolumn{2}{c}{$t$} &
\multicolumn{2}{c}{NPMLE} \\

& & BIC & ICL & BIC & ICL & BIC & ICL \\

\hline

\multirow{2}{*}{$\uppercase\expandafter{\romannumeral1}$}
& RMSE & 16.5744 & 16.5744 & 15.6032 & 15.6032 & $\bm{15.3556}$ & $\bm{15.3556}$ \\
& MAE  & 5.8885  & 5.8885  & 4.8444  & 4.8444  & $\bm{4.5538}$  & $\bm{4.5538}$ \\

\hline

\multirow{2}{*}{$\uppercase\expandafter{\romannumeral2}$}
& RMSE & 15.9982 & 15.9981 & $\bm{15.1549}$ & $\bm{15.1549}$ & 15.3243 & 15.3243 \\
& MAE  & 5.4214  & 5.4212  & $\bm{4.4607}$  & $\bm{4.4607}$  & 4.5884  & 4.5884 \\

\hline

\multirow{2}{*}{$\uppercase\expandafter{\romannumeral3}$}
& RMSE & 15.7991 & 15.7991 & $\bm{15.3167}$ & $\bm{15.3167}$ & 15.3343 & 15.3350 \\
& MAE  & 5.0883  & 5.0883  & $\bm{4.6277}$  & $\bm{4.6277}$  & 4.6354  & 4.6355 \\

\hline

\multirow{2}{*}{$\uppercase\expandafter{\romannumeral4}$}
& RMSE & 16.5176 & 16.5100 & 15.7230 & 15.7230 & $\bm{15.5905}$ & 15.5924 \\
& MAE  & 6.0756  & 6.0704  & 5.0359  & 5.0359  & $\bm{4.8777}$  & 4.8778 \\

\hline

\multirow{2}{*}{$\uppercase\expandafter{\romannumeral5}$}
& RMSE & ${15.2736}$ & $\bm{15.2326}$ & 15.4663 & 15.4663 & 15.3391 & 15.3316 \\
& MAE  & 4.8240  & 4.7890  & 4.7821  & 4.7821  & ${4.5714}$ & $\bm{4.5667}$ \\

\hline

\multirow{2}{*}{$\uppercase\expandafter{\romannumeral6}$}
& RMSE & 16.0105 & 16.0131 & 18.3474 & 18.3474 & ${15.2701}$ & $\bm{15.2691}$ \\
& MAE  & 5.4116  & 5.4196  & 7.8875  & 7.8875  & ${4.6027}$  & $\bm{4.5988}$ \\

\hline
\end{tabular}
\end{table}

The NPMLE estimator attains the smallest or near-smallest prediction errors across the non-Gaussian settings and shows a clear advantage in the more complex contamination scenarios. This advantage is especially evident in Case~\uppercase\expandafter{\romannumeral6}, where the proposed semiparametric error specification provides additional protection against error distributions involving both heavy tails and mixture contamination. These results demonstrate that the NPMLE approach maintains stable predictive performance under distributional misspecification without sacrificing efficiency in the Gaussian case.

\section{Real Data Analysis} \label{sec:real}

\subsection{Tone perception}

The tone perception data originate from the psychoacoustic experiment reported by \citet{cohen1980inharmonic}. The sample consists of 150 observations obtained from repeated tuning trials performed by a single musician. Each observation records two variables: \textit{stretchratio}, which controls the inharmonic spacing of overtones, and \textit{tuned}, representing the musician’s adjustment of the perceived octave.
The experiment compares pure fundamental tones with electronically synthesized overtones. A stretch ratio of 2.0 corresponds to a harmonic structure consistent with conventional definite-pitched instruments, whereas a tuned value of 2.0 indicates uniform octave tuning across stretch ratios. In each trial, the musician adjusted an auxiliary tone to match the octave above the fundamental, producing paired measurements that reflect perceptual tuning behavior under varying harmonic conditions.
This dataset has frequently served as a benchmark example in studies of robust mixture regression models and related clustering methods; see, for example, \citet{bai2012robust}, \citet{song2014robust}, \citet{chamroukhi2016robust}, and \citet{oh2023merging}.

We analyze tone perception data to assess the robustness of competing MoE models to mild contamination. In addition to the original data, we construct an outlier-contaminated version by artificially introducing three outlying observations at 
\[
(\text{stretchratio}, \text{tuned})
=
(1.5,4),\ (2.5,5),\ (2,3.5).
\]
This setup allows us to assess how sensitive each method is to localized departures from the assumed model.
Table~\ref{tab:tone_selection} summarizes the BIC and ICL values under both data settings. For the original data, all methods favor a two-component structure. After adding outliers, the Normal model selects a three-component solution, suggesting that the contamination induces spurious cluster formation. In contrast, both the $t$ and NPMLE models continue to prefer two components, indicating resistance to overfitting caused by outliers.

\begin{table}[ht]
\centering
\caption{BIC and ICL values for the tone perception data under original and outlier-added settings. Boldface indicates the minimum value within each setting.}
\label{tab:tone_selection}

\begin{tabular}{c cc cc cc}
\hline\hline
\multicolumn{7}{c}{\textbf{Original data}} \\
\hline
\multirow{2}{*}{$K$} &
\multicolumn{2}{c}{Normal} &
\multicolumn{2}{c}{$t$} &
\multicolumn{2}{c}{NPMLE} \\
& BIC & ICL & BIC & ICL & BIC & ICL \\
\hline

1 & -3.73 & -3.73 & -6.14  & -6.14 & -36.19 & -36.19 \\
2 & $\bm{-245.60}$ & $\bm{-215.61}$ & $\bm{-403.76}$ & $\bm{-373.61}$ & $\bm{-345.83}$ & $\bm{-303.32}$ \\
3 & -245.19 & -197.91 & -379.67 & -339.13 & -301.37 & -264.88 \\

\hline
\multicolumn{7}{c}{\textbf{Outlier-added data}} \\
\hline

\multirow{2}{*}{$K$} &
\multicolumn{2}{c}{Normal} &
\multicolumn{2}{c}{$t$} &
\multicolumn{2}{c}{NPMLE} \\
& BIC & ICL & BIC & ICL & BIC & ICL \\
\hline

1 & 154.23 & 154.23 & 115.01 & 115.01 & 42.39 & 42.39 \\
2 & -133.13 & -115.03 & $\bm{-362.71}$ & $\bm{-331.88}$ & $\bm{-295.44}$ & $\bm{-251.17}$ \\
3 & $\bm{-354.41}$ & $\bm{-330.71}$ & -358.15 & -328.99 & -266.64 & -230.60 \\

\hline
\end{tabular}
\end{table}

Table~\ref{tab:tone_param} summarizes how each estimator responds to the introduction of outliers. Since the MoE model assigns observations to experts probabilistically, the reported mixing proportion for the $k$th expert is computed as the sample average of the fitted posterior membership probabilities,
$\hat{\pi}_k = \frac{1}{n}\sum_{i=1}^n \hat z_{ik}$.
For the Normal model, the contaminated data lead to a qualitative change in model structure: an additional expert component is selected, and the mixing proportions are redistributed. This structural change is accompanied by noticeable shifts in the regression coefficients, indicating that the Gaussian specification attempts to accommodate the outliers by forming a separate cluster. 
In contrast, both the $t$-based and NPMLE estimators retain the original two-component structure after contamination. Their estimated mixing proportions remain comparable to those obtained from the original data, suggesting that the outliers are absorbed without inducing spurious cluster formation. The difference becomes more pronounced at the level of regression parameters. While the $t$ model exhibits minor adjustments in coefficient estimates, the NPMLE estimator shows almost no change in the estimated $\betab = (\beta_0, \beta_1)^{\top}$ parameters. This near-invariance indicates strong robustness of the semiparametric estimator with respect to the expert-specific regression structure.

\begin{table}[ht]
\centering
\caption{Estimated mixing proportions and regression coefficients under original and outlier-added tone perception data.}
\label{tab:tone_param}

\begin{tabular}{c ccc cc cc}
\hline\hline

\multicolumn{8}{c}{\textbf{Original data}} \\
\hline
\multirow{2}{*}{Parameter} &
\multicolumn{3}{c}{Normal} &
\multicolumn{2}{c}{$t$} &
\multicolumn{2}{c}{NPMLE} \\

& Expert 1 & Expert 2 & --
& Expert 1 & Expert 2
& Expert 1 & Expert 2 \\
\hline

$\pi$
& 0.717 & 0.282 & --
& 0.586 & 0.413
& 0.554 & 0.445 \\

$\beta_0$
& 1.913 & -0.029 & --
& 1.960 & 0.002
& 1.921 & 0.009 \\

$\beta_1$
& 0.043 & 0.995 & --
& 0.025 & 0.999
& 0.039 & 0.996 \\

\hline
\multicolumn{8}{c}{\textbf{Outlier-added data}} \\
\hline
\multirow{2}{*}{Parameter} &
\multicolumn{3}{c}{Normal} &
\multicolumn{2}{c}{$t$} &
\multicolumn{2}{c}{NPMLE} \\

& Expert 1 & Expert 2 & Expert 3
& Expert 1 & Expert 2
& Expert 1 & Expert 2 \\
\hline

$\pi$
& 0.548 & 0.341 & 0.109
& 0.602 & 0.397
& 0.571 & 0.428 \\

$\beta_0$
& 1.924 & 0.003 & 1.169
& 1.963 & 0.002
& 1.921 & 0.009 \\

$\beta_1$
& 0.038 & 0.998 & 0.611
& 0.024 & 0.999
& 0.039 & 0.996 \\

\hline
\end{tabular}
\end{table}

\begin{figure}[http]
\centering

\subfigure[Normal (Original)]{
\includegraphics[width=0.425\linewidth]{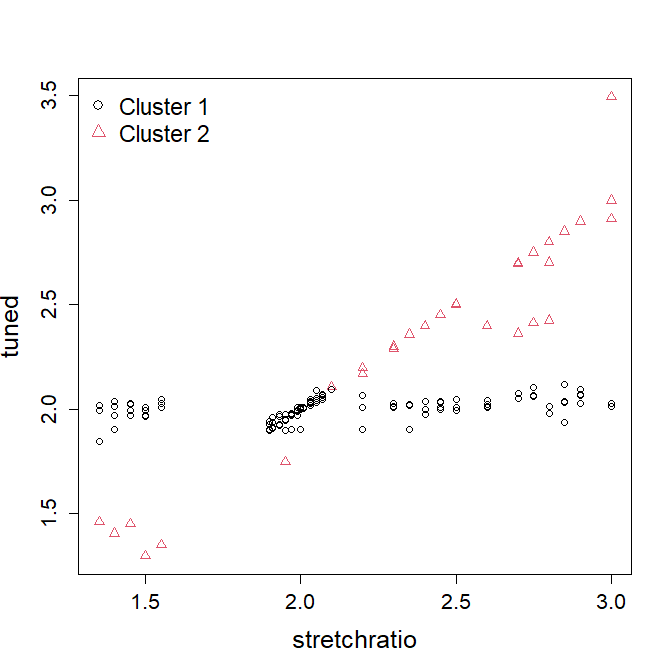}
}\hfill
\subfigure[Normal (Outlier-added)]{
\includegraphics[width=0.425\linewidth]{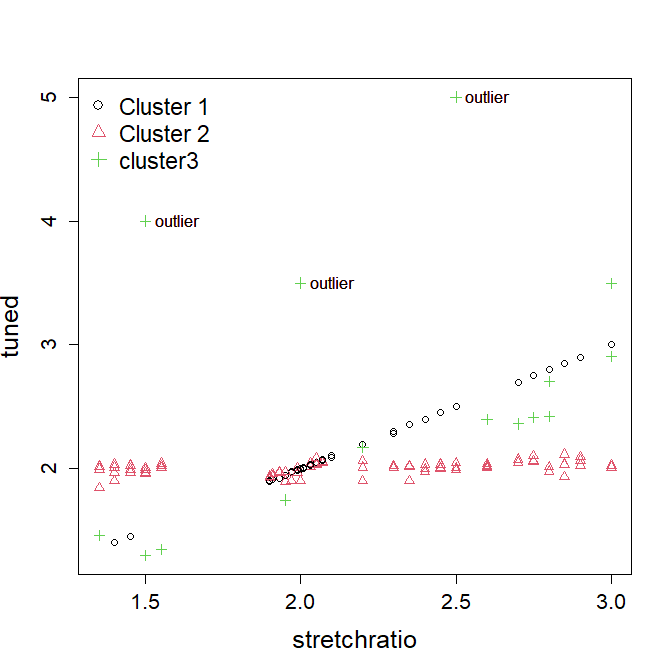}
}

\subfigure[$t$ (Original)]{
\includegraphics[width=0.425\linewidth]{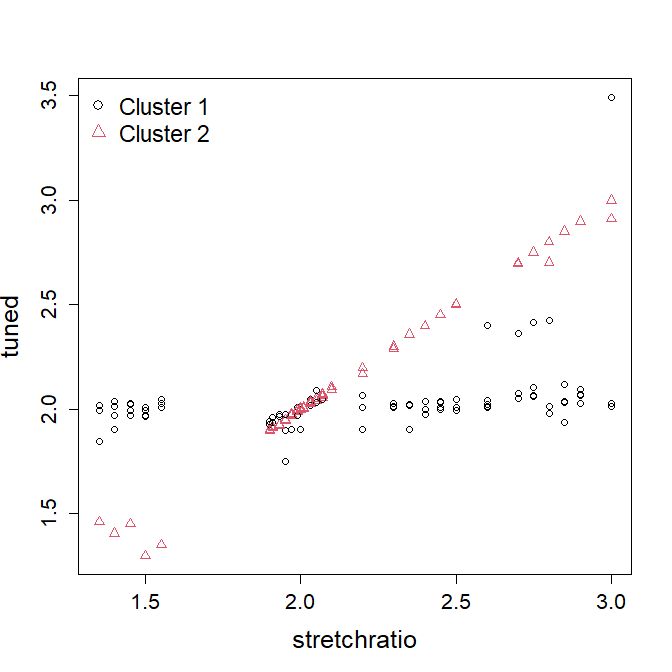}
}\hfill
\subfigure[$t$ (Outlier-added)]{
\includegraphics[width=0.425\linewidth]{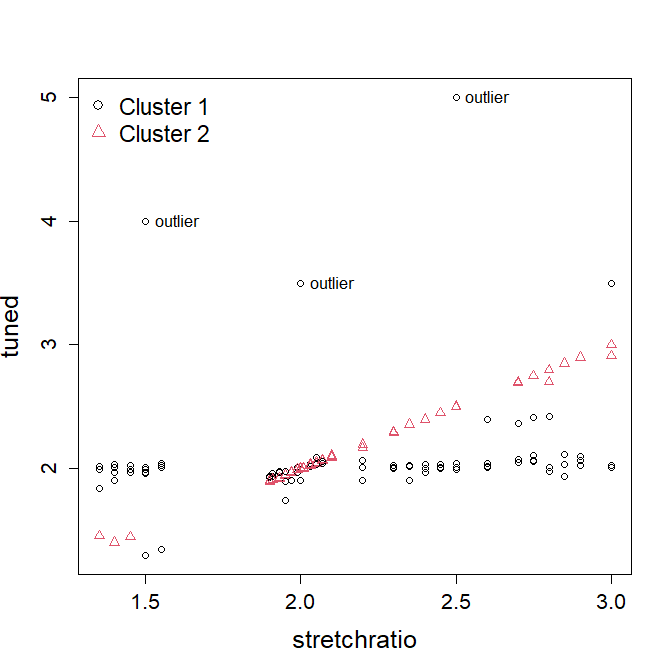}
}

\subfigure[NPMLE (Original)]{
\includegraphics[width=0.425\linewidth]{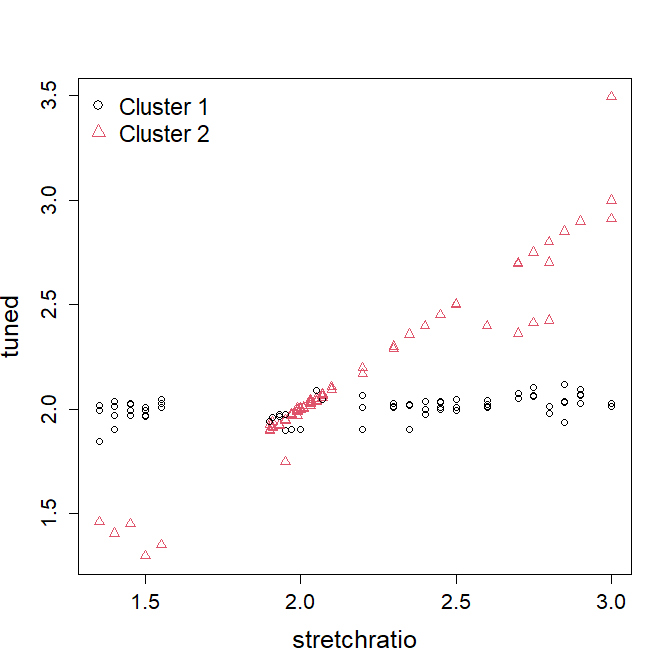}
}\hfill
\subfigure[NPMLE (Outlier-added)]{
\includegraphics[width=0.425\linewidth]{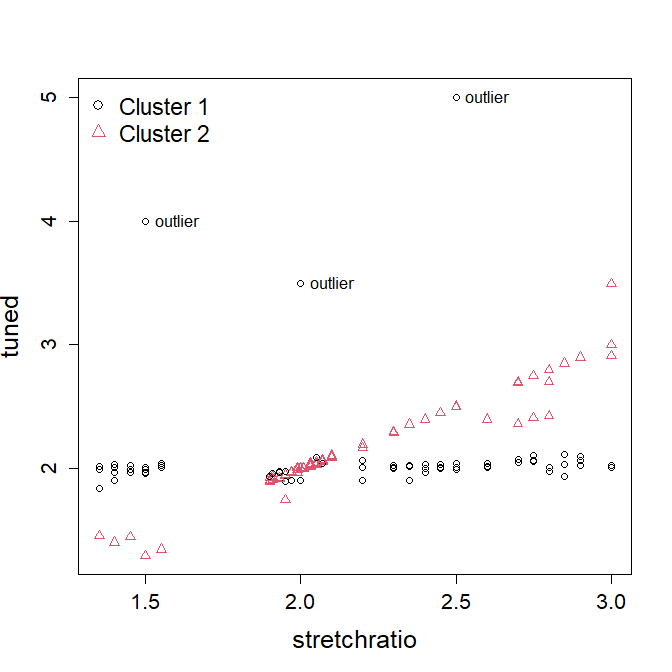}
}

\caption{Clustering comparison between original (left column) and outlier-added data (right column) for Tone perception data}
\label{fig:tone_compare}
\end{figure}

The clustering assignments displayed in Figure~\ref{fig:tone_compare} reinforce these findings. Under the Normal model, the outliers substantially modify the partition of the data, producing visibly different cluster boundaries. By comparison, the $t$ and NPMLE partitions are nearly indistinguishable between the original and contaminated settings. This stability is reflected quantitatively in the proportion of label changes between the two data sets,
\[
74.6\% \text{ (Normal)}, \quad
2.0\% \text{ ($t$)}, \quad
0.0\% \text{ (NPMLE)},
\]
which confirms that the Normal-based estimator is highly sensitive to localized contamination, whereas the $t$ and especially the NPMLE estimator preserve the latent cluster structure. Taken together, the numerical and graphical evidence demonstrate that the proposed semiparametric approach maintains both model complexity and parameter estimates in the presence of outliers.

\subsection{Korea National Health and Nutrition Examination Survey}

As a second real-data illustration, we use data from the Korea National Health and Nutrition Examination Survey (KNHANES), a nationwide, cross-sectional health and nutrition survey that provides representative information on health status, health-related behaviors, dietary intake, and socioeconomic characteristics of the Korean population. KNHANES has been used to study associations between dietary intake patterns and metabolic health outcomes based on 24-hour dietary recall data, demonstrating its relevance for analyzing diet-related health variables \citep{jeong2023higher}. Specifically, we analyze data from the 2021 cycle of KNHANES. After excluding observations with missing values in the response or covariates, the final analytic sample consisted of 21,607 observations.

In this analysis, the response variable is the Healthy Eating Index (HEI), which summarizes overall dietary quality, with larger values indicating better dietary quality. Prior to model fitting, HEI and all continuous covariates were standardized to have mean zero and unit variance, while categorical covariates were represented by dummy variables. The covariates include demographic, socioeconomic, anthropometric, biochemical, and nutritional variables. Specifically, we consider gender, age, household income, education level, body mass index (BMI), fasting glucose, total cholesterol, total energy intake, water intake, protein intake, fat intake, and sodium intake. Household income and education level are treated as categorical variables, with the lowest income group and elementary school or lower education level used as reference categories, respectively. The nutritional variables characterize dietary intake profiles, whereas BMI, fasting glucose, and total cholesterol reflect metabolic health status. The combination of continuous and categorical predictors, together with the large sample size, provides a useful setting for evaluating whether the proposed MoE model can identify heterogeneous regression relationships between HEI and individual-level demographic, socioeconomic, dietary, and metabolic characteristics.

Table~\ref{tab:knhanes_selection} summarizes model selection results based on BIC and ICL for $K=1,2,3$. The Normal and $t$ models attain their smallest BIC values at $K=3$, whereas the NPMLE model attains its smallest BIC value at $K=2$. In contrast, ICL consistently selects the two-component model across all three approaches. This discrepancy is consistent with the known behavior of mixture model selection criteria: BIC tends to favor improved likelihood fit, whereas ICL additionally penalizes classification uncertainty and therefore favors models with clearer component separation. Since the purpose of this real-data analysis is to obtain an interpretable latent structure in the HEI regression relationship, and because ICL consistently supports $K=2$, we proceed with the two-expert model.

\begin{table}[ht]
\centering
\caption{BIC and ICL values for the 2021 KNHANES data (Boldface indicates the minimum value within each criterion).}
\label{tab:knhanes_selection}

\begin{tabular}{c cc cc cc}
\hline\hline
\multirow{2}{*}{$K$} &
\multicolumn{2}{c}{Normal} &
\multicolumn{2}{c}{$t$} &
\multicolumn{2}{c}{NPMLE} \\
& BIC & ICL & BIC & ICL & BIC & ICL \\
\hline

1 & 5458.54 & 5458.54 & 5521.45 & 5521.45 & 5489.07 & 5489.07 \\
2 & 5151.32 & $\bm{5288.55}$ & 5154.79 & $\bm{5322.91}$ & $\bm{5202.86}$ & $\bm{5372.25}$ \\
3 & $\bm{5142.95}$ & 5449.40 & $\bm{5131.31}$ & 5427.24 & 5209.96 & 5613.87 \\

\hline
\end{tabular}
\end{table}

Estimated mixing proportions and regression coefficients for $K=2$ are reported in Table~\ref{tab:knhanes_param}. 
Across the three approaches, the estimated mixing proportions indicate that the sample is divided into two moderately sized latent subgroups. The regression coefficients suggest that the association between HEI and the covariates differs across experts. In particular, gender, age, household income, education level, and dietary intake variables show distinct coefficient patterns between the two experts, indicating that the associations between dietary quality and the covariates are heterogeneous across individuals. The coefficients for total energy intake, water intake, protein intake, fat intake, and sodium intake further suggest that nutritional profiles are associated with HEI differently across the latent subgroups.

\begin{table}[ht]
\centering
\caption{Estimated mixing proportions and regression coefficients for the 2021 KNHANES data.}
\label{tab:knhanes_param}

\begin{tabular}{c cc cc cc}
\hline\hline
Parameter
& \multicolumn{2}{c}{Normal}
& \multicolumn{2}{c}{$t$}
& \multicolumn{2}{c}{NPMLE} \\

& Expert 1 & Expert 2 
& Expert 1 & Expert 2 
& Expert 1 & Expert 2  \\
\hline

$\pi$
& 0.431 & 0.568  
& 0.361 & 0.638  
& 0.462 & 0.537  \\

Intercept
& 0.617 & 0.142 
& 0.841 & 0.159 
& 0.442 & 0.244 \\

Gender (Female)
& 0.453 & 0.136
& 0.479 & 0.133
& 0.476 & 0.106 \\

Age
& 0.389 & 0.410 
& 0.381 & 0.415 
& 0.401 & 0.393 \\

Income (Lower-middle)
& 0.114 & -0.052
& 0.050 & 0.001
& 0.109 & -0.056 \\

Income (Upper-middle)
& 0.043 & 0.026
& 0.065 & 0.036
& 0.053 & 0.035 \\

Income (High)
& 0.120 & -0.042
& 0.106 & -0.006
& 0.151 & -0.071 \\

Education (Middle)
& -0.054 & -0.060
& 0.057 & -0.144
& 0.002 & -0.132 \\

Education (High)
& -0.138 & 0.075
& -0.153 & 0.033
& -0.100 & 0.041 \\

Education (≥ College)
& -0.170 & 0.076
& -0.172 & 0.020
& -0.133 & 0.036 \\

BMI
& -0.054 & -0.032
& -0.046 & -0.043
& -0.044 & -0.040 \\

Glucose
& -0.048 & -0.074
& -0.052 & -0.029
& -0.037 & -0.082 \\

Cholesterol
& -0.072 & -0.032
& -0.091 & -0.020
& -0.072 & -0.035 \\

Energy
& -0.179 & 0.091
& -0.263 & 0.122
& -0.187 & 0.070 \\

Water
& -0.001 & 0.187
& -0.017 & 0.185
& 0.041 & 0.170 \\

Protein
& 1.422 & 0.222
& 1.442 & 0.235
& 1.489 & 0.192 \\

Fat
& 0.416 & -0.489
& 0.734 & -0.509
& 0.193 & -0.477 \\

Sodium
& -0.243 & -0.175
& -0.248 & -0.171
& -0.234 & -0.163 \\

\hline
\end{tabular}
\end{table}

Figure~\ref{fig:knhanes_continuous} displays the clustering results for selected continuous covariates that show clear differences between the two estimated clusters. The boxplots indicate that Cluster 2 tends to consist of younger individuals than Cluster 1 while exhibiting higher overall intake levels for total energy, water, protein, fat, and sodium. By comparison, Cluster 1 is characterized by relatively older individuals with lower intake levels across the selected nutritional variables. The separation is especially visible for energy, protein, fat, and sodium intake, where the medians and upper tails are shifted upward in Cluster 2. These patterns suggest that the estimated two-cluster structure is distinguished by a combined age--nutrition profile, providing an interpretable summary of latent heterogeneity in the KNHANES data.

\begin{figure}[p]
\centering

\subfigure[Age]{
\includegraphics[width=0.375\linewidth]{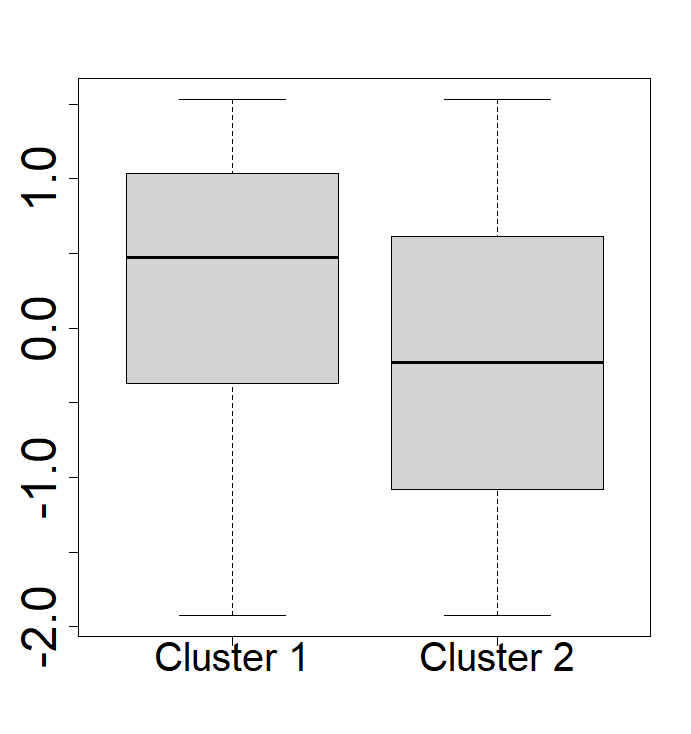}
}\hfill
\subfigure[Energy]{
\includegraphics[width=0.375\linewidth]{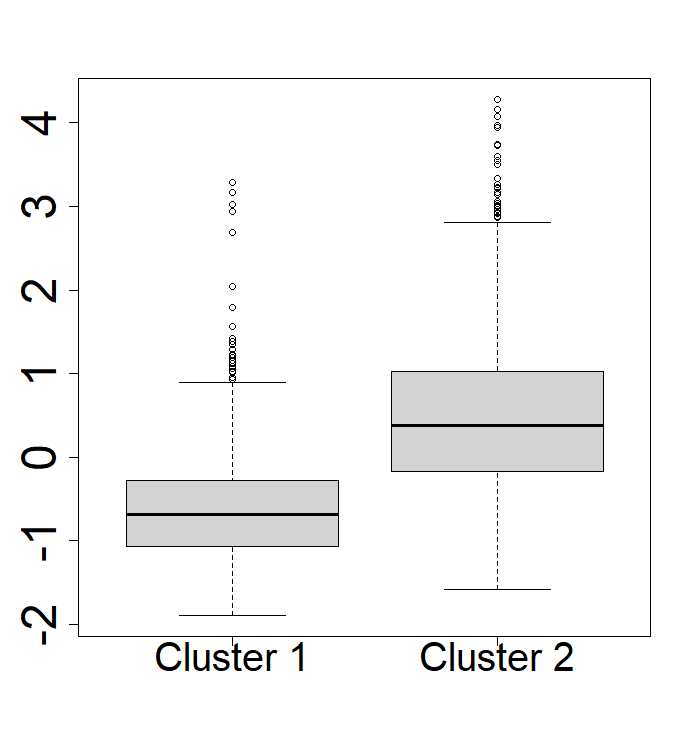}
}\hfill
\subfigure[Water]{
\includegraphics[width=0.375\linewidth]{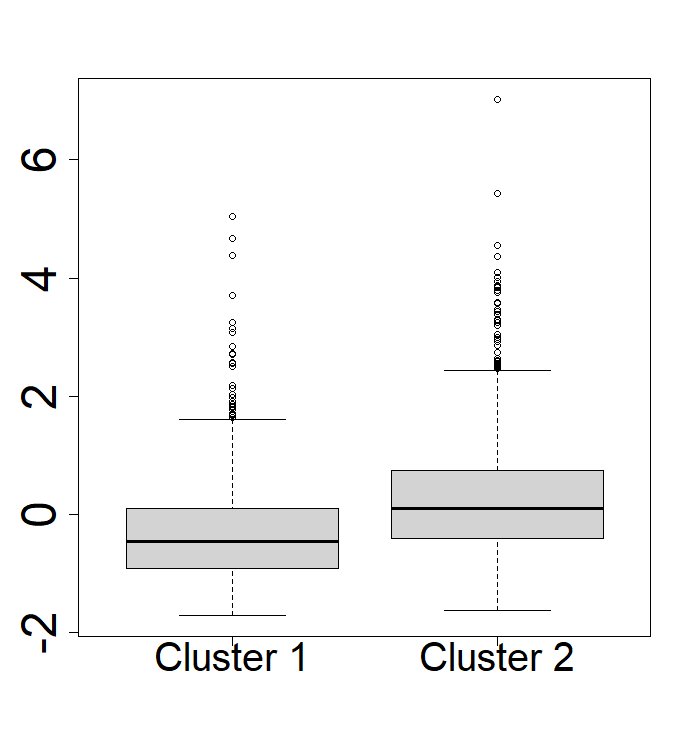}
}\hfill
\subfigure[Protein]{
\includegraphics[width=0.375\linewidth]{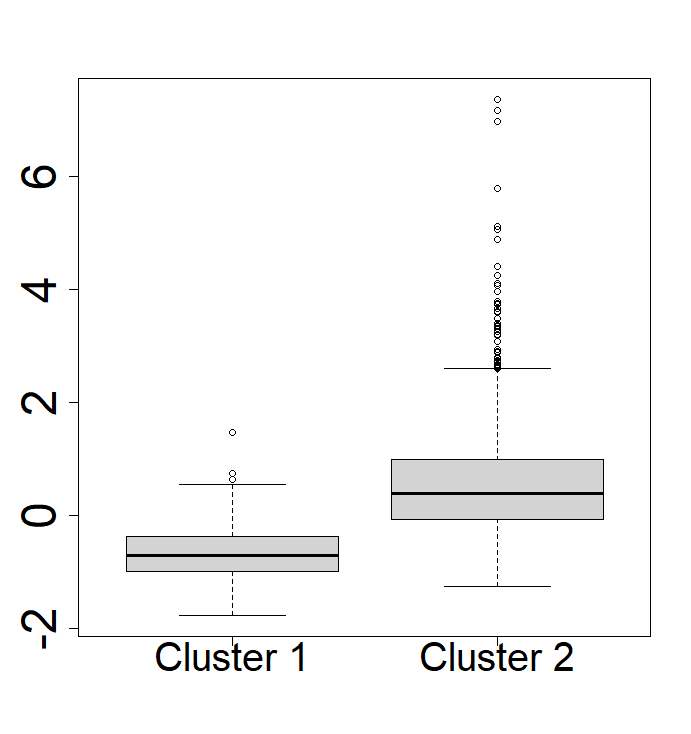}
}\hfill
\subfigure[Fat]{
\includegraphics[width=0.375\linewidth]{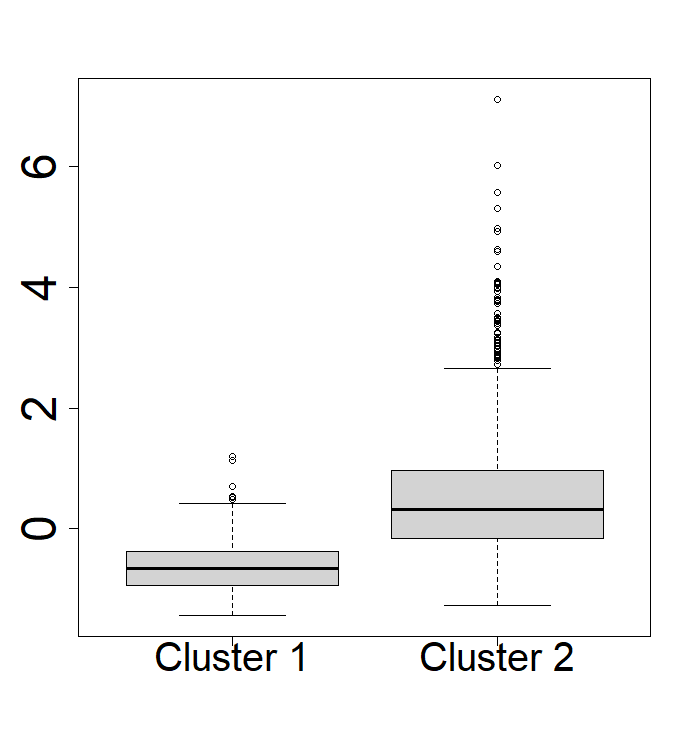}
}\hfill
\subfigure[Sodium]{
\includegraphics[width=0.375\linewidth]{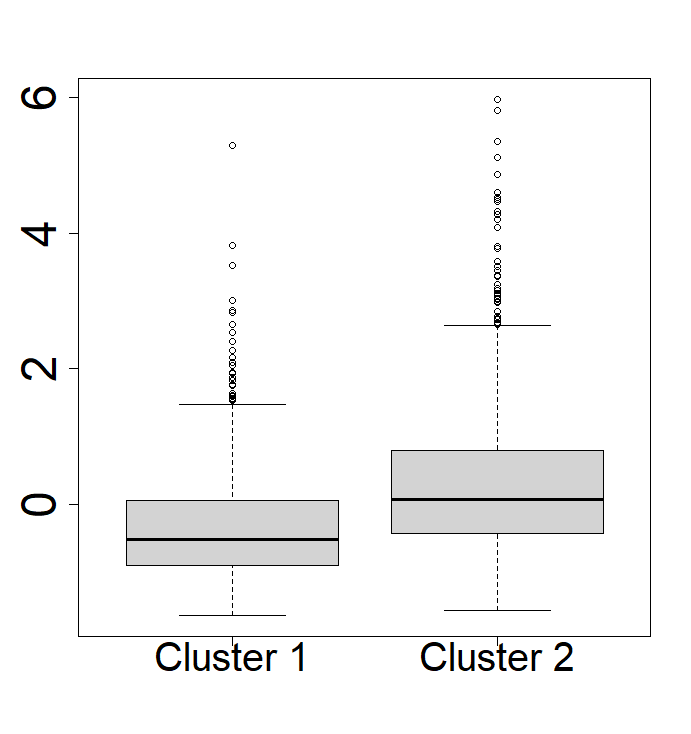}
}
\caption{
NPMLE clustering results for the 2021 KNHANES data based on continuous covariates: age, total energy intake, water intake, protein intake, fat intake, and sodium intake.
}
\label{fig:knhanes_continuous}
\end{figure}

\begin{figure}[p]
\centering

\subfigure[Gender]{
\includegraphics[width=0.8\linewidth]{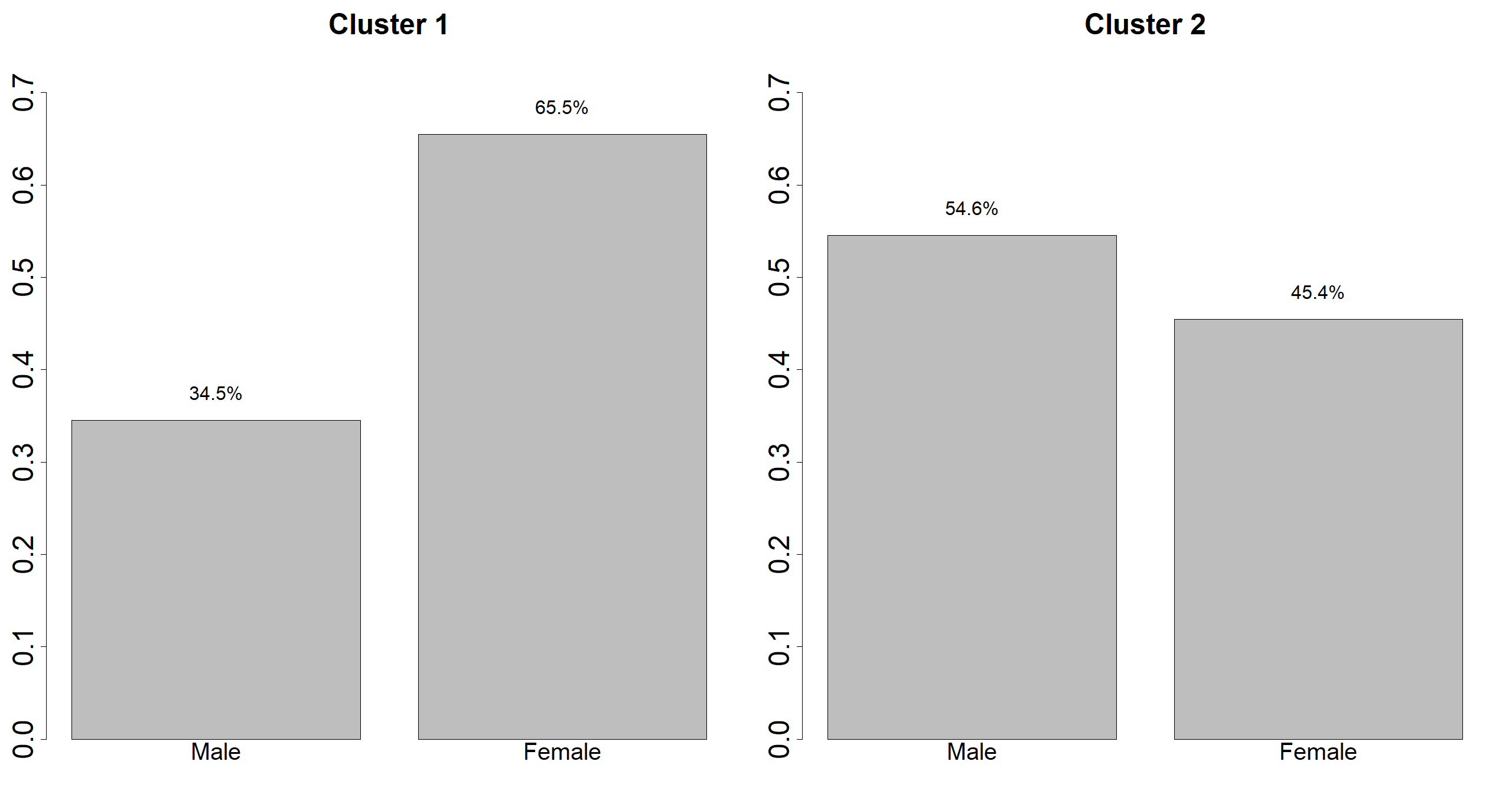}
}\hfill
\subfigure[Income]{
\includegraphics[width=0.8\linewidth]{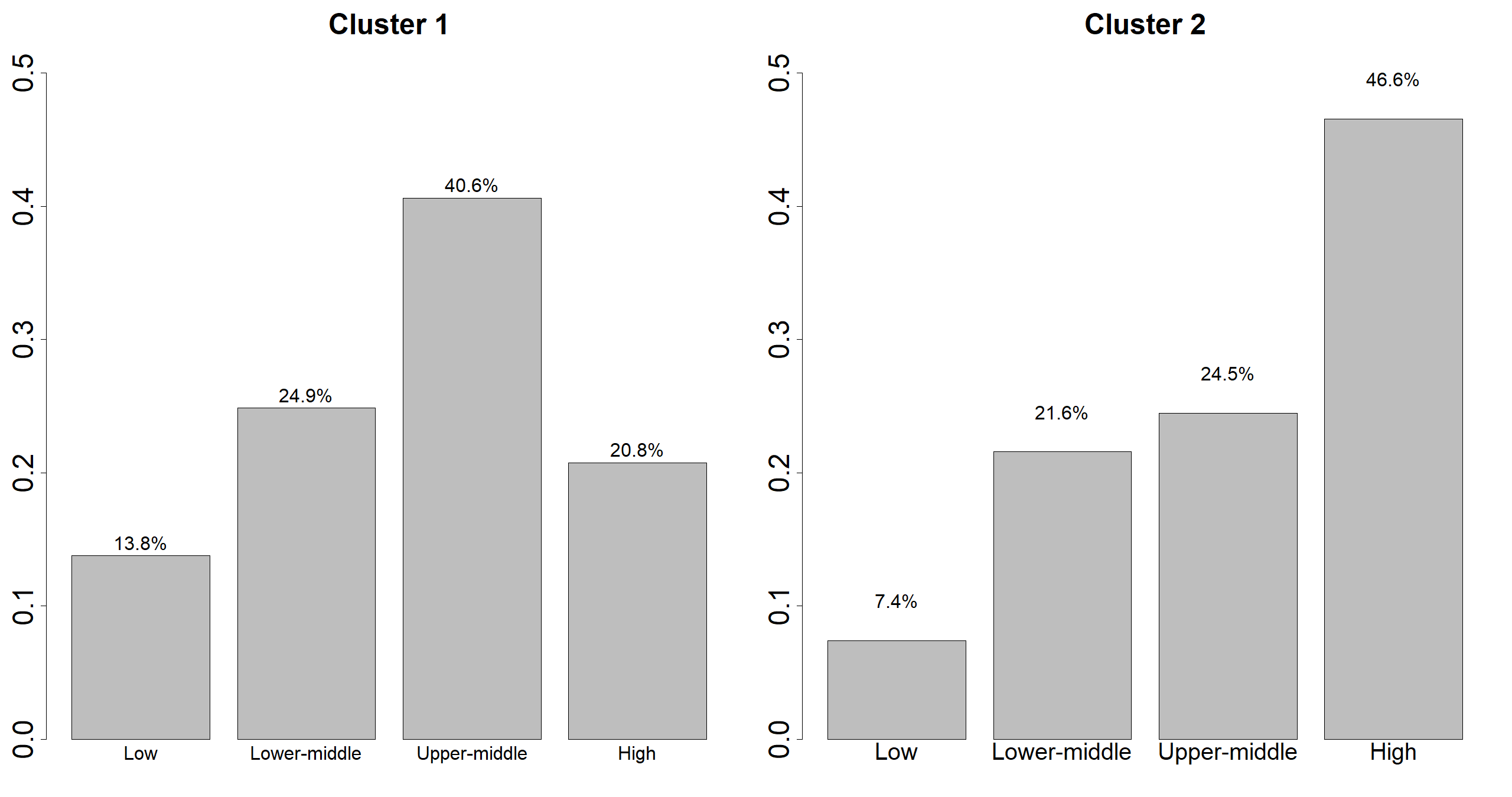}
}\hfill
\subfigure[Education]{
\includegraphics[width=0.8\linewidth]{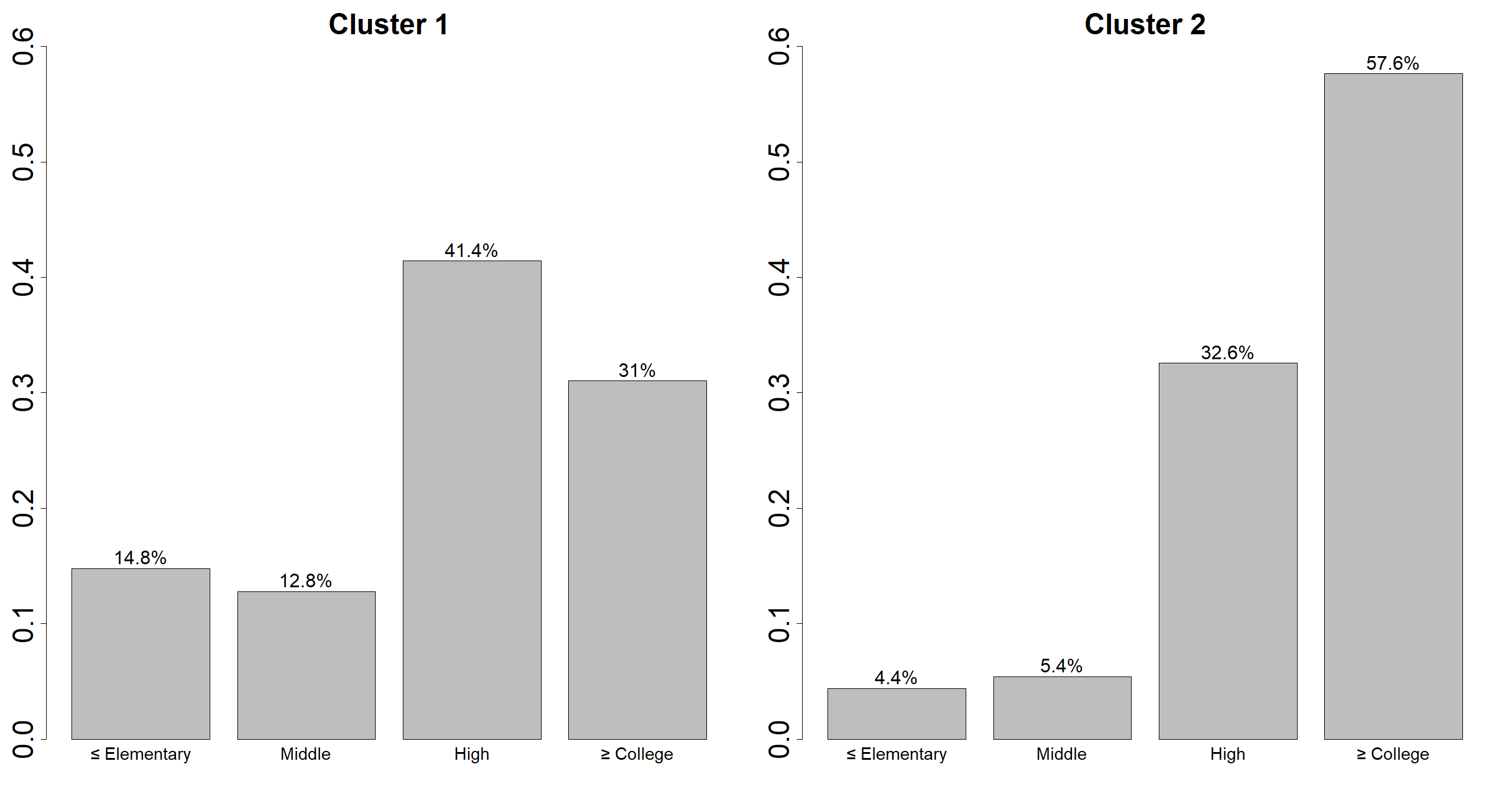}
}
\caption{
NPMLE clustering results for the 2021 KNHANES data based on categorical covariates: gender, household income, and education level.
}
\label{fig:knhanes_categorical}
\end{figure}

Figure~\ref{fig:knhanes_categorical} summarizes the clustering results for categorical covariates, including gender, household income, and education level. The gender distribution shows that Cluster 1 has a higher proportion of females, whereas Cluster 2 has a higher proportion of males. Clear differences are also observed in household income and education level. Cluster 1 is more concentrated in the lower-middle and upper-middle income categories, while Cluster 2 has a substantially larger proportion of individuals in the high-income category. A similar pattern appears for education. Cluster 1 includes relatively larger proportions of individuals with elementary, middle, or high school education, whereas Cluster 2 is dominated by individuals with college-level education or higher. These patterns suggest that the two estimated clusters differ not only in dietary intake profiles but also in demographic and socioeconomic composition, with Cluster 2 representing a younger, higher-intake subgroup with higher socioeconomic status.

Predictive performance is evaluated using 10-fold cross-validation, with RMSE and MAE summarized in Figure~\ref{fig:knhanes_cv}. The boxplots compare the Normal, $t$, and NPMLE models in terms of out-of-sample prediction accuracy for HEI. The NPMLE model yields the lowest median RMSE and MAE among the three methods, indicating improved central predictive accuracy relative to the parametric alternatives. It also exhibits smaller dispersion across folds for both criteria, suggesting more stable predictive performance under possible non-Gaussian error behavior or influential observations in the KNHANES data.

\begin{figure}[h]
\centering

\subfigure[RMSE]{
\includegraphics[width=0.435\linewidth]{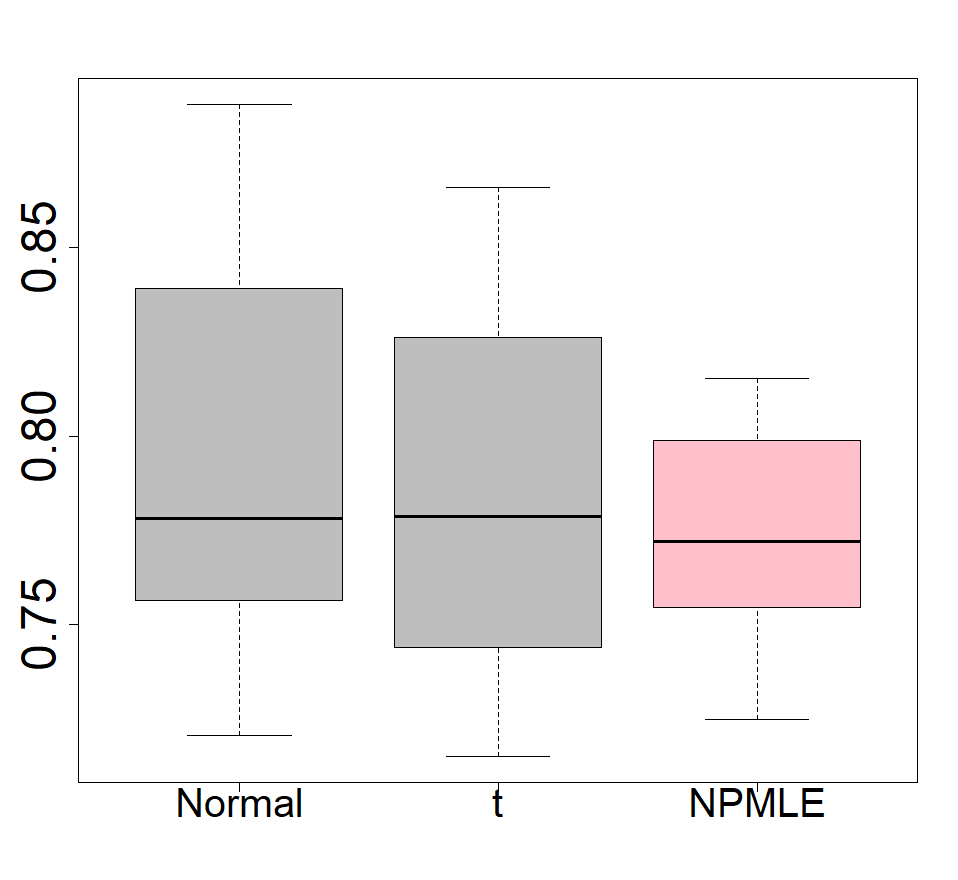}
}\hfill
\subfigure[MAE]{
\includegraphics[width=0.435\linewidth]{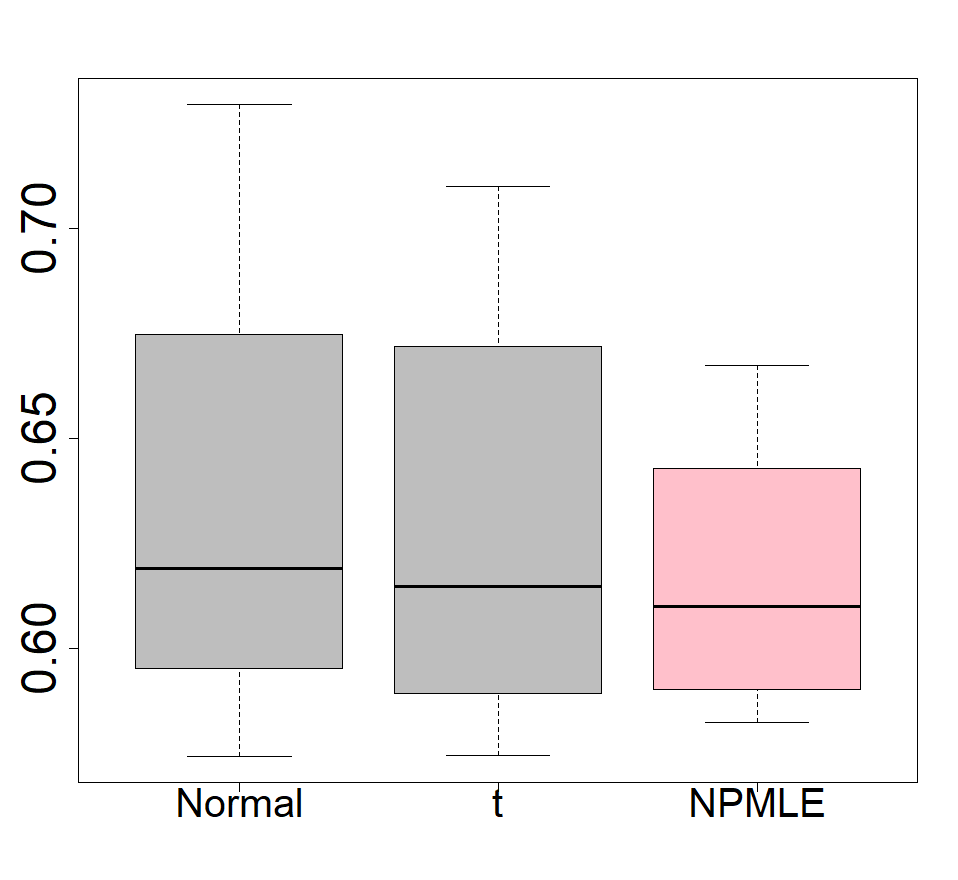}
}

\caption{
Predictive performance on the 2021 KNHANES data under 10-fold cross-validation. 
}
\label{fig:knhanes_cv}
\end{figure}

\section{Discussion} \label{sec:disc}

We propose a semiparametric MoE method in which expert error distributions are modeled using nonparametric Gaussian scale mixtures estimated via NPMLE. This formulation relaxes parametric assumptions within the Gaussian scale mixture class while preserving the structural advantages of the classical MoE architecture. Theoretical analysis establishes identifiability, consistency, and monotonicity of the estimation algorithm, providing a rigorous foundation for semiparametric expert modeling. 
From an applied perspective, the proposed approach is particularly well-suited to regression and clustering tasks in which the underlying error structure is unknown or exhibits heavy-tailed behavior. By allowing each expert component to adapt its error distribution to the data, the model reduces sensitivity to misspecification while retaining interpretability through the gating–expert decomposition.

Several limitations warrant further investigation. The increased flexibility of the semiparametric formulation introduces additional computational complexity relative to parametric MoE models, particularly when estimating nonparametric mixing distributions. Efficient implementation and scalable optimization, therefore, remain important practical considerations. 
Future work may explore extensions to multivariate responses, structured covariates, high-dimensional settings, and censored/missing data \citep{lachos2017finite,alencar2022finite}. The integration of regularization techniques and alternative nonparametric representations could further improve scalability and interpretability. A broader theoretical investigation of asymptotic efficiency and convergence behavior is another important direction for future research.






\bibliographystyle{apa} 

\bibliography{references}

\section*{Appendix A: Proof of Lemma \ref{lem:gsm-ident}}

Assume that, for all $x\in\mathbb{R}$,
\[
\sum_{k=1}^{K} \pi_k 
\int \frac{1}{\sigma}
\phi\!\left(\frac{x-\mu_k}{\sigma}\right)\, dG_k(\sigma)
=
\sum_{j=1}^{\tilde K} \tilde\pi_j 
\int \frac{1}{\sigma}
\phi\!\left(\frac{x-\tilde\mu_j}{\sigma}\right)\, d\tilde G_j(\sigma).
\]
Let $\Psi(t)$ and $\tilde\Psi(t)$ denote the corresponding characteristic functions. Taking characteristic functions on both sides gives, for all \(t\in\mathbb R\),
\[
\sum_{k=1}^{K}\pi_k e^{i\mu_k t}\Psi_{G_k}(t)
=
\sum_{j=1}^{\tilde K}\tilde\pi_j e^{i\tilde\mu_j t}
\Psi_{\tilde G_j}(t),
\]
where 
\[
\Psi_G(t)=\int \exp\!\left(-\frac12\sigma^2t^2\right)dG(\sigma).
\]

Let \(H_1,\ldots,H_B\) be the distinct scale mixing distributions appearing
among
\[
G_1,\ldots,G_K,\tilde G_1,\ldots,\tilde G_{\tilde K}.
\]
By Assumption~\textup{(A3)}, these distributions can be relabelled so that
\[
\frac{\Psi_{H_{b'}}(t)}{\Psi_{H_b}(t)}\to0,
\qquad |t|\to\infty,\qquad 1\le b<b'\le B .
\]

For each \(b=1,\ldots,B\), define
\[
I_b=\{k:G_k=H_b\},
\qquad
\tilde I_b=\{j:\tilde G_j=H_b\}.
\]
Then the characteristic-function identity can be written as
\[
\sum_{b=1}^{B}\Psi_{H_b}(t)P_b(t)=0,
\]
where
\[
P_b(t)
=
\sum_{k\in I_b}\pi_k e^{i\mu_k t}
-
\sum_{j\in \tilde I_b}\tilde\pi_j e^{i\tilde\mu_j t}.
\]
Each \(P_b\) is bounded because
\[
|P_b(t)|
\le
\sum_{k\in I_b}\pi_k
+
\sum_{j\in \tilde I_b}\tilde\pi_j
\le 2.
\]

Dividing the preceding identity by \(\Psi_{H_1}(t)\) gives
\[
P_1(t)
+
\sum_{b=2}^{B}
\frac{\Psi_{H_b}(t)}{\Psi_{H_1}(t)}P_b(t)
=0.
\]
Hence,
\[
|P_1(t)|
\le
\sum_{b=2}^{B}
\left|
\frac{\Psi_{H_b}(t)}{\Psi_{H_1}(t)}
\right|
|P_b(t)|
\to0,
\qquad |t|\to\infty,
\]
by the asymptotic ordering of \(H_1,\ldots,H_B\) and the boundedness of
\(P_b\).

Write
\[
P_1(t)=\sum_{\nu\in\mathcal M_1}a_{1\nu}e^{i\nu t},
\]
where \(\mathcal M_1\) is the finite set of distinct locations appearing in
the first block. Since \(P_1(t)\to0\) as \(|t|\to\infty\) and \(P_1\) is
bounded,
\[
0
=
\lim_{T\to\infty}
\frac{1}{2T}\int_{-T}^{T}|P_1(t)|^2dt
=
\sum_{\nu\in\mathcal M_1}|a_{1\nu}|^2.
\]
Thus \(a_{1\nu}=0\) for all \(\nu\in\mathcal M_1\), and hence
\(P_1\equiv0\).
Removing the first block and repeating the same argument successively with
the remaining leading blocks yields
\[
P_b\equiv0,\qquad b=1,\ldots,B.
\]

Therefore, for each \(b\),
\[
\sum_{k\in I_b}\pi_k e^{i\mu_k t}
=
\sum_{j\in \tilde I_b}\tilde\pi_j e^{i\tilde\mu_j t},
\qquad t\in\mathbb R.
\]
Equivalently, the finite positive measures
\[
Q_b=\sum_{k\in I_b}\pi_k\delta_{\mu_k},
\qquad
\tilde Q_b=\sum_{j\in \tilde I_b}\tilde\pi_j\delta_{\tilde\mu_j}
\]
have the same characteristic function. By uniqueness of characteristic
functions for finite measures,
\[
Q_b=\tilde Q_b.
\]
Since the locations \(\mu_1,\ldots,\mu_K\) are pairwise distinct and
\(\tilde\mu_1,\ldots,\tilde\mu_{\tilde K}\) are pairwise distinct, equality of
these discrete measures implies that, within each block \(b\), the atoms and
their corresponding masses coincide. Hence there exists a bijection
\(\tau_b:I_b\to\tilde I_b\) such that, for every \(k\in I_b\),
\[
\mu_k=\tilde\mu_{\tau_b(k)},
\qquad
\pi_k=\tilde\pi_{\tau_b(k)}.
\]
Moreover, by the definition of \(I_b\) and \(\tilde I_b\),
\[
G_k=H_b=\tilde G_{\tau_b(k)}.
\]

Combining the bijections \(\tau_b\), \(b=1,\ldots,B\), gives a permutation
\(\tau\) of \(\{1,\ldots,K\}\). In particular,
\[
\tilde K=K
\]
and, for all \(k=1,\ldots,K\),
\[
\pi_k=\tilde\pi_{\tau(k)},\qquad
\mu_k=\tilde\mu_{\tau(k)},\qquad
G_k=\tilde G_{\tau(k)}.
\]
\section*{Appendix B: Proof of Theorem \ref{thm:ident}}

For notational simplicity, we present the proof for the case in
which no intercept terms are included in the expert and gating
networks. The same argument applies when intercepts are included.
Indeed, in that case, write
\[
m_k(\ub)
=
\beta_{k0}+\ub^\top\betab_k,
\qquad
\eta_k(\ub)
=
\alpha_{k0}+\ub^\top\alphab_k.
\]
For two distinct augmented regression coefficient vectors
\[
(\beta_{k0},\betab_k^\top)^\top
\quad\text{and}\quad
(\beta_{k'0},\betab_{k'}^\top)^\top,
\]
the set
\[
\left\{
\ub\in\mathbb R^p:
(\beta_{k0}-\beta_{k'0})
+
\ub^\top(\betab_k-\betab_{k'})
=0
\right\}
\]
is either an affine hyperplane or the empty set, and hence has
Lebesgue measure zero. Moreover, equality of two affine functions
on a nonempty open set implies equality of both their intercept
and slope coefficients. The same argument applies to the affine
log-odds functions in the gating network. Therefore, the proof
with intercept terms proceeds identically.

Assume that there exists an alternative representation with \(\tilde K\)
components such that, for all
\((y,\ub)\in\mathbb R\times\mathbb R^p\),
\begin{align}
\label{7}
p(y\mid \ub)
&=\sum_{k=1}^K \pi(\ub;\alphab_k)
\int \frac{1}{\sigma}
\phi\!\left(\frac{y-\ub^\top\betab_k}{\sigma}\right)\,dG_k(\sigma) \notag\\
&=\sum_{j=1}^{\tilde K} \pi(\ub;\tilde\alphab_j)
\int \frac{1}{\sigma}
\phi\!\left(\frac{y-\ub^\top\tilde\betab_j}{\sigma}\right)\,d\tilde G_j(\sigma),
\end{align}
where \(\tilde\betab_1,\ldots,\tilde\betab_{\tilde K}\) are pairwise
distinct. For convenience, define
\[
f_{\mathrm{GSM}}(u;G)
=
\int \frac{1}{\sigma}
\phi\!\left(\frac{u}{\sigma}\right)\,dG(\sigma),
\qquad
m_k(\ub)=\ub^\top\betab_k,
\qquad
\tilde m_j(\ub)=\ub^\top\tilde\betab_j.
\]

For any \(k,k' \in \{1,\ldots,K\}\) with \(k\neq k'\), since
\(\betab_k\neq \betab_{k'}\) by Assumption~\textup{(A2)}, the set
\[
\{\ub\in\mathbb R^p:\ \ub^\top(\betab_k-\betab_{k'})=0\}
\]
is a hyperplane and hence has Lebesgue measure zero. Likewise, for any
\(j,j' \in \{1,\ldots,\tilde K\}\) with \(j\neq j'\), the set
\[
\{\ub\in\mathbb R^p:\ \ub^\top(\tilde\betab_j-\tilde\betab_{j'})=0\}
\]
is a hyperplane and hence has Lebesgue measure zero.

By Assumption~\textup{(A1)}, the support of \(\ub\) contains a nonempty open set
\(\bm U\subset\mathbb R^p\). Hence, we can choose
\[
\ub^\star \in \bm U\setminus
\Bigg(
\bigcup_{\substack{k\neq k'\\ 1\le k,k'\le K}}
\{\ub:\ub^\top(\betab_k-\betab_{k'})=0\}
\;\cup\!
\bigcup_{\substack{j\neq j'\\ 1\le j,j'\le \tilde K}}
\{\ub:\ub^\top(\tilde\betab_j-\tilde\betab_{j'})=0\}
\Bigg).
\]
For this choice, the collections \(\{m_k(\ub^\star)\}_{k=1}^K\) and
\(\{\tilde m_j(\ub^\star)\}_{j=1}^{\tilde K}\) are pairwise distinct within
each family.

Fixing \(\ub=\ub^\star\) in \eqref{7}, we obtain, for all \(y\in\mathbb R\),
\[
\sum_{k=1}^K \pi_k^\star\, f_{\mathrm{GSM}}(y-m_k^\star;G_k)
=
\sum_{j=1}^{\tilde K} \tilde\pi_j^\star\,
f_{\mathrm{GSM}}(y-\tilde m_j^\star;\tilde G_j),
\]
where
\[
\pi_k^\star=\pi(\ub^\star;\alphab_k),
\qquad
\tilde\pi_j^\star=\pi(\ub^\star;\tilde\alphab_j),
\qquad
m_k^\star=m_k(\ub^\star),
\qquad
\tilde m_j^\star=\tilde m_j(\ub^\star).
\]
Under the multinomial-logit gating network, \(\pi_k^\star>0\) and
\(\sum_k \pi_k^\star=1\), and similarly for \(\tilde\pi_j^\star\).
Moreover, by the choice of \(\ub^\star\), the component locations are pairwise
distinct within each representation. Since Assumption~\textup{(A3)} holds,
Lemma~\ref{lem:gsm-ident} applies. Hence, \(\tilde K=K\) and there exists a
permutation \(\tau_{\ub^\star}\) of \(\{1,\ldots,K\}\) such that
\begin{equation}
\label{8}
m_k(\ub^\star)
=
\tilde m_{\tau_{\ub^\star}(k)}(\ub^\star),
\qquad
\pi(\ub^\star;\alphab_k)
=
\pi(\ub^\star;\tilde\alphab_{\tau_{\ub^\star}(k)}),
\qquad
G_k=\tilde G_{\tau_{\ub^\star}(k)}.
\end{equation}

Let \(\tau=\tau_{\ub^\star}\). Since
\(m_1(\ub^\star),\ldots,m_K(\ub^\star)\) are pairwise distinct, the matched
locations
\[
m_k(\ub^\star)
=
\tilde m_{\tau(k)}(\ub^\star),
\qquad k=1,\ldots,K,
\]
are separated. Hence, by continuity, there exists a nonempty open neighborhood
\(V\subset\bm U\) of \(\ub^\star\) such that, for every \(\ub\in V\), the
locations \(m_1(\ub),\ldots,m_K(\ub)\) and
\(\tilde m_1(\ub),\ldots,\tilde m_K(\ub)\) remain pairwise distinct, and the
matching induced by Lemma~\ref{lem:gsm-ident} is still given by the same
permutation \(\tau\). Therefore, for all \(\ub\in V\) and \(k=1,\ldots,K\),
\[
m_k(\ub)=\tilde m_{\tau(k)}(\ub),
\qquad
\pi(\ub;\alphab_k)=\pi(\ub;\tilde\alphab_{\tau(k)}),
\qquad
G_k=\tilde G_{\tau(k)}.
\]

Consequently,
\[
\ub^\top\betab_k
=
\ub^\top\tilde\betab_{\tau(k)},
\qquad
\ub\in V,\quad k=1,\ldots,K.
\]
Since both sides are linear in \(\ub\) and agree on the nonempty open set \(V\),
it follows that
\[
\betab_k=\tilde\betab_{\tau(k)},
\qquad k=1,\ldots,K.
\]

Finally, since the mixing proportions agree on \(V\),
\[
\pi(\ub;\alphab_k)
=
\pi(\ub;\tilde\alphab_{\tau(k)}),
\qquad
\ub\in V,\quad k=1,\ldots,K.
\]
In particular,
\[
\pi(\ub;\alphab_K)
=
\pi(\ub;\tilde\alphab_{\tau(K)}),
\qquad
\ub\in V.
\]
Under the multinomial-logit gating network with the baseline constraint
\(\alphab_K=0\), we have
\[
\log\frac{\pi(\ub;\alphab_k)}
{\pi(\ub;\alphab_K)}
=
\ub^\top\alphab_k.
\]
For the alternative representation,
\[
\log
\frac{\pi(\ub;\tilde\alphab_{\tau(k)})}
{\pi(\ub;\tilde\alphab_{\tau(K)})}
=
\ub^\top
\left(
\tilde\alphab_{\tau(k)}
-
\tilde\alphab_{\tau(K)}
\right).
\]
Therefore,
\[
\ub^\top\alphab_k
=
\ub^\top
\left(
\tilde\alphab_{\tau(k)}
-
\tilde\alphab_{\tau(K)}
\right),
\qquad
\ub\in V.
\]
Since \(V\) contains a nonempty open set, it follows that
\[
\alphab_k
=
\tilde\alphab_{\tau(k)}
-
\tilde\alphab_{\tau(K)},
\qquad
k=1,\ldots,K.
\]

\section*{Appendix C: Proof of Theorem \ref{thm:consistency}}

Under model~\eqref{eq:spmoe},
\[
p_{\Thetab,\Gb}(y\mid\xb)
=
\sum_{k=1}^K
\pi(\xb;\alphab_k)
\int_{\ell}^{\infty}
\frac{1}{\sigma}
\phi\!\left(
\frac{y-\xb^\top\betab_k}{\sigma}
\right)
\,dG_k(\sigma).
\]
For
$\xb=(1,\ub^\top)^\top$, define the joint density of $(Y,\Ub)$ by
\[
q_{\Thetab,\Gb}(y,\ub)
=
p_{\Thetab,\Gb}(y\mid\xb)h_0(\ub).
\]
Let $\thetab=(\Thetab,\Gb)$ and
$\widetilde\thetab=(\widetilde\Thetab,\widetilde\Gb)$, and define
\[
\begin{aligned}
d(\thetab,\widetilde\thetab)
={}&
\|\alphab-\widetilde\alphab\|_1
+
\|\betab-\widetilde\betab\|_1\\
&+
\sum_{k=1}^K
\int_{\ell}^{\infty}
|G_k(\sigma)-\widetilde G_k(\sigma)|
e^{-\sigma}\,d\tau(\sigma).
\end{aligned}
\]

The joint density $q_{\Thetab,\Gb}$ satisfies Assumptions 1--3 of
\citet{kiefer1956consistency} by standard measurability and continuity
arguments, and Assumption 4 follows from
Theorem~\ref{thm:ident}, up to label permutation. Moreover, since
each $G_k$ is supported on $[\ell,\infty)$,
\[
q_{\Thetab,\Gb}(y,\ub)
\le
\frac{\|h_0\|_{\infty}}{\ell\sqrt{2\pi}}
\]
uniformly in $(y,\ub)$ and $(\Thetab,\Gb)$. Thus, to verify
Assumption 5, it suffices to show that
\[
\mathbb E_0
\left[
-\log q_{\Thetab_0,\Gb_0}(Y,\Ub)
\right]
<\infty.
\]

Fix $k^*\in\{1,\ldots,K\}$. Since
\[
p_{\Thetab_0,\Gb_0}(y\mid\xb)
\ge
\pi(\xb;\alphab_{0k^*})f_{0k^*}(y\mid\xb),
\]
where
\[
f_{0k}(y\mid\xb)
=
\int_{\ell}^{\infty}
\frac{1}{\sqrt{2\pi}\sigma}
\exp\!\left\{
-\frac{(y-\xb^\top\betab_{0k})^2}{2\sigma^2}
\right\}
\,dG_{0k}(\sigma),
\]
we have
\[
\begin{aligned}
-\log q_{\Thetab_0,\Gb_0}(Y,\Ub)
\le{}&
-\log\pi(\Xb;\alphab_{0k^*})
-\log f_{0k^*}(Y\mid\Xb)
-\log h_0(\Ub).
\end{aligned}
\]

There exist finite constants $C_0$ and $C_1$
such that
\[
\begin{aligned}
-\log\pi(\Xb;\alphab_{0k^*})
&=
\log
\left\{
\sum_{j=1}^K
\exp(\alphab_{0j}^\top\Xb)
\right\}
-
\alphab_{0k^*}^\top\Xb
\\
&\le
C_0+C_1\|\Xb\|.
\end{aligned}
\]
Since
\[
\mathbb E_0\|\Xb\|^2
=
1+\mathbb E_0\|\Ub\|^2
<\infty,
\]
it follows that
\[
\mathbb E_0
\left[
-\log\pi(\Xb;\alphab_{0k^*})
\right]
<\infty.
\]

By Jensen's inequality and Assumption~\textup{(A4)},
\[
\begin{aligned}
-\log f_{0k^*}(y\mid\xb)
&\le
\frac12\log(2\pi)
+
\int_{\ell}^{\infty}
\log\sigma\,dG_{0k^*}(\sigma)
+
\frac{(y-\xb^\top\betab_{0k^*})^2}{2\ell^2}\\
&\le
\frac12\log(2\pi)+M
+
\frac{(y-\xb^\top\betab_{0k^*})^2}{2\ell^2}.
\end{aligned}
\]
Since $\mathbb E_0Y^2<\infty$ and $\mathbb E_0\|\Xb\|^2<\infty$,
\[
\mathbb E_0
\left[
-\log f_{0k^*}(Y\mid\Xb)
\right]
<\infty.
\]

Finally,
\[
\mathbb E_0[-\log h_0(\Ub)]<\infty
\]
by assumption. Hence,
\[
\mathbb E_0
\left[
-\log q_{\Thetab_0,\Gb_0}(Y,\Ub)
\right]
<\infty.
\]
Therefore, Assumption 5 holds. Hence, after a suitable relabeling of
the estimated components,
\[
d\left(
(\widehat{\Thetab}_n,\widehat{\Gb}_n),
(\Thetab_0,\Gb_0)
\right)
\xrightarrow{p}0.
\]

\end{document}